\documentclass{amsart}
\usepackage[T1]{fontenc}   
\usepackage[utf8]{inputenc}
\usepackage[english]{babel} 
\usepackage{amsmath,amsfonts,amssymb,amsthm,mathtools}
\usepackage{amsmath}
\usepackage{mathabx}
\usepackage{thmtools}
\usepackage{thm-restate}
\usepackage{mathrsfs}
\usepackage{url}
\usepackage{color}
\usepackage{hyperref}
\usepackage{stmaryrd}
\usepackage{enumitem}

\usepackage{tikz}
\usepackage{tikz-cd}
\usetikzlibrary{tikzmark,calc,arrows, shapes, decorations.pathreplacing, positioning, quotes}

\theoremstyle{theorem}
\newtheorem{thm}{Theorem}
\newtheorem{assumption}[thm]{Assumption}{\bfseries}{\itshape}
\newtheorem{prop}[thm]{Proposition}
\newtheorem{proposition}[thm]{Proposition}

\newtheorem{heur}[thm]{Heuristic}
\newtheorem{lemma}[thm]{Lemma}
\newtheorem{conjec}[thm]{Conjecture}
\theoremstyle{definition}
\newtheorem{definition}[thm]{Definition}
\theoremstyle{remark}
\newtheorem{remark}[thm]{Remark}
\newtheorem{fact}[thm]{Fact}
 \renewcommand{\leq}{\leqslant} 
\renewcommand{\geq}{\geqslant} 

\newcommand{\F}{\ensuremath{\mathbb{F}}}
\newcommand{\Fq}{\ensuremath{\mathbb{F}_q}}
\newcommand{\fq}{\ensuremath{\mathbb{F}_q}}
\newcommand{\Fqm}{\ensuremath{\mathbb{F}_{q^m}}}
\newcommand{\fqm}{\ensuremath{\mathbb{F}_{q^m}}}

\newcommand{\Z}{\mathbb{Z}}

\newcommand{\mat}[1]{\ensuremath{\boldsymbol{#1}}}
\newcommand{\code}[1]{\ensuremath{\mathscr{#1}}}

\newcommand{\AC}{\code{A}}
\newcommand{\BC}{\code{B}}
\newcommand{\CC}{\code{C}}
\newcommand{\DC}{\code{D}}

\newcommand{\XC}{\code{X}}

\newcommand{\Gm}{\mat{G}}

\renewcommand{\Im}{\mat{I}}
\newcommand{\Mm}{\mat{M}}       
\newcommand{\Pm}{\mat{P}}

\newcommand{\Vm}{\mat{V}}
\newcommand{\Xm}{\mat{X}}
\newcommand{\Ym}{\mat{Y}}
\newcommand{\Wm}{\mat{W}}

\newcommand{\zerom}{\mat{0}}

\newcommand{\av}{\mat{a}}
\newcommand{\bv}{\mat{b}}
\newcommand{\cv}{\mat{c}}
\newcommand{\dv}{\mat{d}}

\newcommand{\uv}{\mat{u}}
\newcommand{\vv}{\mat{v}}
\newcommand{\xv}{{\mat{x}}}
\newcommand{\yv}{{\mat{y}}}

\newcommand{\cB}{{\mathcal{B}}}
\newcommand{\cE}{{\mathcal{E}}}
\newcommand{\cM}{{\mathcal{M}}}
\newcommand{\cS}{{\mathcal{S}}}
\newcommand{\cV}{{\mathcal{V}}}
\newcommand{\cU}{{\mathcal{U}}}

\newcommand{\GRS}[3]{\text{\bf GRS}_{#1}(#2,#3)}
\newcommand{\Alt}[3]{\code{A}_{#1}(#2, #3)}
\newcommand{\Goppa}[2]{\code{G}(#1, #2)}
\newcommand{\sh}[2]{\mathbf{Sh}_{#1}\left(#2\right)}

\newcommand{\pu}[2]{\mathbf{Pct}_{#1}\left(#2\right)}
\newcommand{\puv}[2]{#1_{\widecheck{#2}}}
\newcommand{\sss}[1]{\stackrel{\mathbf{Sh}_{#1}}{\supseteq}}

\newcommand{\Tr}[1]{\textrm{Tr}\left( #1 \right)}
\newcommand{\trace}[1]{\Tr{#1}} \newcommand{\Trm}[1]{\textrm{Tr}_{\Fqm/\Fq}\left( #1 \right)}
\newcommand{\dual}[1]{\ensuremath{{#1}^{\perp}}}
\newcommand{\trsp}[1]{{#1}^{\top}}

\newcommand{\starp}[2]{{#1} \star {#2}}

\newcommand{\sq}[1]{#1^{\star 2}}
\newcommand{\sqb}[1]{\left(#1\right)^{\star 2}}

\newcommand{\cond}[2]{\text{\bf Cond}\left(#1, #2\right)}

\newcommand{\Ical}{\mathcal{I}}

\newcommand{\rank}{\mathbf{Rank}}
\newcommand{\eqdef}{\stackrel{\text{def}}{=}}
\newcommand{\ie}{\textit{i.e.}\,}
\newcommand{\Span}[2]{\left\langle \, #1 \, \right\rangle_{#2}}
\newcommand{\Fqspan}[1]{\left\langle \, #1 \, \right\rangle_{\Fq}}
\newcommand{\Fqmspan}[1]{\left\langle \, #1 \, \right\rangle_{\Fqm}}
\newcommand{\floor}[1]{\left\lfloor #1 \right\rfloor}

\newcommand{\Iintv}[2]{\llbracket #1 , #2 \rrbracket}
\newcommand{\card}[1]{\lvert #1 \rvert}

\newcommand{\bigO}[1]{\mathcal{O}\left(#1\right)}
\newcommand{\macaulay}[1]{\cM^{\text{acaulay}}\left(#1\right)}

\DeclareMathOperator{\LM}{LM}

\title[Polynomial time attack on high rate random alternant codes]{Polynomial time key-recovery attack on high rate random alternant codes}

\author{Magali Bardet}
 \address{Univ Rouen Normandie, Normandie Univ, LITIS UR 4108, 76000 Rouen, France}
 \email{magali.bardet@univ-rouen.fr}
 \author{Rocco Mora}
 \address{Inria, 2 rue Simone Iff, 75012 Paris, France\\
   Sorbonne Universit\'es, UPMC Univ Paris 06}
 \email{rocco.mora@inria.fr}
 \author{Jean--Pierre Tillich}
 \address{Inria, 2 rue Simone Iff, 75012 Paris, France}
 \email{jean-pierre.tillich@inria.fr}

\begin{document}

\maketitle

\begin{abstract}
A long standing open question is whether the distinguisher of high rate alternant codes or Goppa codes  \cite{FGOPT11} can be turned into an algorithm recovering
the algebraic structure of such codes from the mere knowledge of an arbitrary generator matrix of it. This would allow to break the McEliece scheme  as
soon as the code rate is large enough and would break all instances of the CFS signature scheme. We give for the first time a positive answer for this problem when the code is {\em a generic alternant code}
and when the code field size $q$ is small : $q \in \{2,3\}$ and for {\em all} regime of other parameters for which the aforementioned distinguisher works. This breakthrough has been obtained by two different ingredients : \\
(i) a way of using code shortening and the component-wise product of codes to derive
from the original alternant code a sequence of alternant codes of decreasing degree up to getting an alternant code of degree $3$
(with a multiplier and support related to those of the original alternant code);\\
 (ii) an original Gröbner basis approach which takes into account the non standard constraints on the multiplier and support of an alternant code which recovers in polynomial time the relevant algebraic structure of an alternant code of degree $3$ from the mere knowledge of a basis for it.
\end{abstract}
\section{Introduction}
\label{sec:intro}

\subsection*{The McEliece scheme.}
The McEliece encryption scheme \cite{M78}, which dates back to 1978 and which is almost as old as RSA \cite{RSA78}, is a code-based cryptosystem built upon the family of binary Goppa codes. We will denote this system by McEliece-binary Goppa from now on since we will be interested in variations of the McEliece cryptosystem obtained by changing the underlying code family.
 It is equipped with very fast encryption and decryption algorithms and has very small ciphertexts but large public keysize. Contrarily to RSA which is broken by quantum computers \cite{S94a}, 
it is also widely viewed as a viable quantum safe cryptosystem. A variation of this public key cryptosystem intended to be IND-CCA secure and an associated key exchange protocol \cite{ABCCGLMMMNPPPSSSTW20} are now 
in the process of being standardized by the NIST as quantum safe cryptosystems.  
It should be noted that the best quantum algorithm for breaking it \cite{KT17a} has exponential complexity and the corresponding exponent barely improves the exponent of the best 
classical algorithm \cite{BM17} by about $40$ percent.

The attacks on the McEliece scheme can be put in two classes: message-recovery attacks on one hand and key-recovery attacks on the other hand. The first one consists in inverting the McEliece encryption without finding a trapdoor and makes use of generic decoding algorithms. Despite considerable improvements \cite{S88,CC98,MMT11,BJMM12,MO15,BM17}, all these algorithms have exponential complexity. The second one consists in trying to recover the private key or an equivalent private key which allows to decipher as efficiently as a legitimate user. The most efficient attack of this kind was given in \cite{LS01} and essentially consists in trying all Goppa polynomials and all possible supports, where the verification consists in solving a code equivalence problem which is often easy with the help of the support splitting algorithm \cite{S00}. The complexity of the second attack is also exponential with an exponent which is much bigger than the one obtained for the first attack. This is why the first kind of attack is considered as the main threat against McEliece-Goppa, and consequently
 the parameters of all those schemes based on Goppa codes have been chosen to thwart the first attack.

\subsection*{Towards efficient key recovery attacks.}
This picture began to change due on one hand to variations on the original McEliece proposal based on binary Goppa codes of rate close to $\frac{1}{2}$ by either considering very high rate binary Goppa codes for devising signature schemes \cite{CFS01}, or by moving from binary Goppa codes to nonbinary Goppa codes over large alphabets \cite{BLP10,BLP11a}, thus decreasing significantly the extension degree $m$ over which the (secret) support of the Goppa code is defined. Recall here that the Goppa code is defined over some finite field $\fq$ whereas the support of the Goppa code is defined over an extension field $\fqm$ (see Definition \ref{def:Goppa} which follows). On the other hand, more structured versions of the McEliece system also appeared, based on quasi-cyclic codes such as \cite{BCGO09,BIGQUAKE} or on 
quasi-dyadic codes such as \cite{MB09,BLM11,BBBCDGGHKNNPR17}. 

The quasi-cyclic or quasi-dyadic Goppa codes could be attacked by an algebraic modeling \cite{FOPT10,GL09} for the secret key which could be efficiently solved with Gr\"obner bases techniques  because the added structure allowed to reduce drastically the number of unknowns of the algebraic system. By trying to solve the same algebraic system in the case of high rate Goppa codes it was also found that Gr\"obner bases techniques behaved very differently when the system corresponds to a Goppa code instead of a random linear code of the same length and dimension. This approach lead to 
\cite{FGOPT11} that gave a way to distinguish high rate Goppa codes from random codes. 
It is based on the kernel of a linear system related to the aforementioned algebraic system.  Indeed, it was shown to have an unexpectedly high dimension when instantiated with Goppa codes or the more general family of alternant codes rather than with random linear codes. This distinguisher was later on given another interpretation in \cite{MP12}, where it was proved that this dimension is related to the dimension of the square of the dual of the Goppa code. Very recently, \cite{MT21} revisited \cite{FGOPT11} and gave rigorous bounds for the dimensions of the square codes of Goppa or alternant codes and a better insight on the algebraic structure of these squares. Recall here that the component-wise product of codes (also called the Schur product) is defined from the component-wise product of vectors 
$\av=(a_i)_{1 \leq i \leq n}$ and $\bv=(b_i)_{1 \leq i \leq n}$
\[\starp{\av}{\bv}\eqdef(a_1 b_1,\dots,a_n b_n)\]
 by
 \begin{definition}
The  \textit{component-wise product of codes} $\CC,\DC$ over $\F$ with the same length $n$ is defined as
   \[
      \starp{\CC}{\DC}\eqdef \Span{\starp{\cv}{\dv} \mid \cv \in \CC, \dv \in \DC}{\F}.       \]
     If $\CC=\DC$, we call $\sq{\CC}\eqdef\starp{\CC}{\CC}$ the \textit{square code} of $\CC$. 
 \end{definition}

\subsection*{Square code and cryptanalysis.}
Interestingly enough, it was proved in \cite{CGGOT14} that square code considerations could also be used to mount an attack on McEliece or Niederreiter schemes based 
on Generalized Reed-Solomon (GRS) codes. Recall that this scheme was proposed in \cite{N86} and was subsequently broken in \cite{SS92}. Note that when the extension degree of the Goppa code is $1$ (i.e. the support of the Goppa code is defined over the same field as the Goppa code itself), a Goppa code is indeed a GRS code, so a McEliece scheme based on a Goppa code of extension degree $1$ can be attacked with the \cite{SS92} attack. However, this does not seem to generalize to higher extension degrees; i.e. on  McEliece schemes based on Goppa codes in general. The point of the new attack \cite{CGGOT14},  is that it uses arguments on square codes for which there is hope that they could be 
applied to a much broader class of Goppa codes. It is insightful here to recall in a simplified way the attack obtained in \cite[\S5]{CGGOT14}. 

Indeed, GRS codes turn out to display a very peculiar property with respect to the square of codes. It is readily seen that $\dim \sq{\CC} \le \min\left(n,\binom{k +1}{2}\right)$ where $k$ and $n$ are respectively the dimension and length of $\CC$. It turns out that for random codes 
the upper-bound is almost always an equality \cite{CCMZ15}, whereas the situation for GRS codes is completely different: in this case, we namely have 
\begin{equation}
\label{eq:GRSsquaredim}
\dim \sq{\CC} = \min\left(n,2k-1\right).
\end{equation}
 The reason of this peculiar behavior comes from the fact that GRS codes are polynomial evaluation codes. Indeed, recall that a GRS code is given by
\begin{definition}[GRS code]\label{def:GRS}
  Let $\xv=(x_1,\dots,x_n)\in\F^n$ be a vector of pairwise distinct entries and $\yv=(y_1,\dots,y_n)\in\F^n$ a vector of nonzero entries. The $[n, k]$ \textit{generalized Reed-Solomon (GRS) code} with \textit{support} $\xv$ and \textit{multiplier} $\yv$ is
  \[
  \GRS{k}{\xv}{\yv}\eqdef\{(y_1 P(x_1),\dots,y_n P(x_n)) \mid P \in \F[z], \deg P < k\}.
  \]
  \end{definition}
Since the Schur product of two polynomial evaluations of degree $\deg P \leq k-1$ and $\deg Q \leq k-1$ respectively is itself a polynomial evaluation of degree $\deg (P\cdot Q)=\deg P + \deg Q < 2k-1$: 
$(y_i P(x_i))_i \star (y_i Q(x_i))_i = (y_i^2 P\cdot Q(x_i))_i$, it is readily seen that
\begin{equation}
\label{eq:GRSsquareexpression}
\sq{\GRS{k}{\xv}{\yv}} =\GRS{2k-1}{\xv}{\yv\star\yv}, 
\end{equation}
which explains \eqref{eq:GRSsquaredim}. In a sense, the square code construction ``sees'' the polynomial structure of the GRS code. 
A key recovery attack could be mounted as follows. Recall that it amounts here to recover from an arbitrary generator matrix of 
a GRS code $\CC=\GRS{k}{\xv}{\yv}$  a pair $(\xv',\yv')$ satisfying 
$\CC=\GRS{k}{\xv'}{\yv'}$. Let us  define $\CC(i)$ as the subcode of the GRS code $\CC$ given by
\[
\CC(i) = \left\{(y_i P(x_j))_{1 \leq j \leq n}: \deg P < k, \text{ $x_1$ is a zero of order $\geq i$ of $P$}\right\}, 
\]
then
\begin{enumerate}
\item[(i)] $\CC(1)$ can be readily computed from $\CC$ since it is the shortened code of $\CC$ in the first position.
\item[(ii)] We have in general the equality 
\begin{equation}
\label{eq:attackCOT}
\CC(i-1) \star \CC(i+1) = \sq{\CC(i)}
\end{equation}
coming from the fact that the product of two polynomials which have a zero of order $i$ at $x_1$ gives a polynomial with a zero of order $2i$ in $x_1$ and so does the product 
of a polynomial with a zero of order $i-1$ in $x_1$ with a polynomial which has a zero of order $i+1$  at the same place.
\item[(iii)] 
Solving the equation $\XC \star \AC = \BC$ for two known linear codes $\AC$ and $\BC$ amounts to solve a linear system in the case where
$\XC$ is the maximal code satisfying $\XC \star \AC \subseteq \BC$. This is indeed the case here for
$\AC=\CC(i-1)$ and $\BC=\sq{\CC(i)}$. $\XC$ corresponds in such a case to the {\em conductor} of $\AC$ into $\BC$,  which can be computed by solving a linear system (see below).
\end{enumerate}
\begin{definition}
Let $\CC, \DC \subseteq \F^n$ be two codes. The \textit{conductor} of $\CC$ into $\DC$ is
\[
\cond{\CC}{\DC}\eqdef \{\uv \in \Fq^n \mid \starp{\uv}{\CC}\subseteq\DC\},
\]
where $\starp{\uv}{\CC}\eqdef\{\starp{\uv}{\cv} \mid \cv \in \CC\}$.
\end{definition}
\begin{proposition}[{\cite[Lemma 7]{CMP17}}]
  The conductor can be computed with 
\[
\cond{\CC}{\DC} = \dual{(\CC \star \dual{\DC})}.
\]
\end{proposition}
The first two points show that we can therefore
compute $\CC(2)$ in polynomial time, because $\CC(1)$ and $\CC(0)$ are known (the first is the shortened code and the second is the code $\CC$ itself). We can iterate this process and compute 
recursively the decreasing set of codes
$\CC(3)$, $\CC(4), \cdots$  and stop when we get a code of dimension $1$ (i.e. $C(k-1)$) which basically reveals a great deal of information about the multiplier $\yv$. 
It is then straightforward with this approach to finish the attack to recover the whole algebraic structure of $\CC$. Note that we have basically computed a decreasing set of codes
\[
\CC=\CC(0) \supset \CC(1) \supset \CC(2) \cdots \supset \CC(k-1)
\]
that we will call a {\em filtration} in what follows.

This approach works basically like this to attack a McEliece scheme based on GRS codes \cite{CGGOT14}, but interestingly enough it also applies to Wild Goppa codes of extension degree $2$ as shown in  \cite{COT14}. Such schemes were indeed proposed in  \cite{BLP10}. This extension degree  corresponds to the largest extension degree where we can expect the Goppa code to behave differently from a random linear code with respect to the square code dimension. 
Roughly speaking in this case, even if Goppa codes are subfield subcodes of GRS codes, the equality  \eqref{eq:attackCOT} ``almost'' holds and this is sufficient to mount a similar attack.
However as explained in \cite{COT14}, this approach is bound to fail when the extension degree $m$ is bigger than $2$. However as observed in \cite{MP12}, even when $m>2$, the square code $\sq{\CC}$ can also be 
of unusually small dimension when the rate of the Goppa code is close to $1$, but this time not by taking $\CC$ to be the Goppa code itself, but by choosing $\CC$ to be the {\em dual} of the Goppa code. This strongly suggests that an approach similar to \cite{CGGOT14,COT14} could be followed to attack McEliece schemes based on very high rate Goppa codes. Even if the parameters of the McEliece schemes proposed in the literature are never in the regime where the dimension of the square of the dual of the Goppa code behaves differently from a random linear code, there is the notable exception of the code-based signature scheme \cite{CFS01}, which is based in a crucial way on high rate Goppa codes, and which similarly to the McEliece scheme would be broken, if we can recover the unknown support of the Goppa code from an arbitrary generator matrix for it. However, the fact that the dual code is actually the trace code of a GRS code but not a subfield subcode of a GRS code loses a lot of the original polynomial structure and seems to complicate very significantly this approach. This is still an open problem since the problem was explicitly raised in \cite{FGOPT11}. 
 
\subsection*{Our contribution.}
In the present article, we make what we consider to be a significant step in this direction. We will namely show that somewhat unexpectedly, an equality related to \eqref{eq:attackCOT} holds, when taking duals of {\em (generic) high rate alternant codes}, but not when we take Goppa codes. This is extremely surprising because Goppa codes are just alternant codes with a peculiar structure.  To explain the phenomenon we are witnessing, let us recall the definition of an alternant code which is a subfield subcode of a GRS code:
\begin{definition}[alternant code]\label{def:alternant}
    Let $n\le q^m$, for some positive integer $m$. Let $\GRS{r}{\xv}{\yv}$ be the GRS code over $\Fqm$ of dimension $r$ with support $\xv \in \Fqm^n$ and multiplier $\yv\in (\Fqm\setminus\{0\})^n$. The \textit{alternant code} with support $\xv$ and multiplier $\yv$ and \textit{degree} $r$ over $\Fq$ is
    \[
    \Alt{r}{\xv}{\yv}\eqdef \GRS{r}{\xv}{\yv}^\perp \cap \F_q^n= \GRS{n-r}{\xv}{\yv^\perp} \cap \F_q^n,
    \]
    where  \[
  \yv^\perp\eqdef\left(\frac{1}{\pi'_\xv(x_1)y_1},\dots,\frac{1}{\pi'_\xv(x_n)y_n}\right)
  \]
  and $\pi'_\xv$ is the derivative of $\pi_\xv$.
    The integer $m$ is called \textit{extension degree} of the alternant code.
  \end{definition}
  In other words, an alternant code is the dual (over $\Fq$) of a GRS code defined over $\Fqm$. It is also a subfield subcode of a GRS code since the dual of a GRS code is a GRS code : $\GRS{r}{\xv}{\yv}^\perp=\GRS{n-r}{\xv}{\yv^\perp}$, see 
   \cite[Theorem~4, p.~304]{MS86}.
  We also recall that Goppa codes form a particular family of alternant codes
   \begin{definition}[Goppa code]\label{def:Goppa}
    Let $\xv\in\Fqm^n$ be a support vector and $\Gamma\in\Fqm[z]$ a polynomial of degree $r$ such that $\Gamma(x_i)\neq 0$ for all $i \in \{1,\dots,n\}$. The \textit{Goppa code} of degree $r$ with support $\xv$ and \textit{Goppa polynomial} $\Gamma$ is defined as
    \[
    \Goppa{\xv}{\Gamma}\eqdef\Alt{r}{\xv}{\yv},
    \]
    where $\yv\eqdef\left(\frac{1}{\Gamma(x_1)},\dots,\frac{1}{\Gamma(x_n)}\right).$
  \end{definition}
  The very unusual behavior we observe in the case of a generic alternant code $\Alt{r}{\xv}{\yv}$ is that in a certain high rate regime, 
  if we shorten its {\em dual} in one position $i$ and take its square to get $\BC \eqdef\sqb{\sh{i}{\Alt{r}{\xv}{\yv}^\perp}}$, where
  $\sh{i}{\CC}$ denotes the code $\CC$ shortened in position $i$, then the maximal code $\XC$ satisfying
  \[
\XC \star \AC \subseteq \BC
  \]
  is the dual of a certain alternant code of degree $r-1$:
  \[
  \XC=  \Alt{r-1}{\puv{\xv}{i}}{\puv{\yv}{i}(\puv{\xv}{i}-x_i)}^\perp.
  \]
  where $\puv{\xv}{i}$ denotes the vector $\xv$ where we have dropped the index $i$ and $\AC$ is the dual of the shortened alternant code in position $i$, i.e.
  $\left(\sh{i}{\Alt{r}{\xv}{\yv}}\right)^\perp$. We also use the short notation $\puv{\yv}{i}(\puv{\xv}{i}-x_i)$ for $\puv{\yv}{i}\star (x_j-x_i)_{j\ne i}$. Note that this code is actually the dual of an alternant code since $\left(\sh{i}{\Alt{r}{\xv}{\yv}}\right)^\perp= \Alt{r}{\puv{\xv}{i}}{\puv{\yv}{i}}^\perp$ (see Proposition \ref{prop:sh_alt}). In other words
  \begin{equation}
  \label{eq:fundamental}
  \cond{\Alt{r}{\puv{\xv}{i}}{\puv{\yv}{i}}}{\sqb{\sh{i}{\Alt{r}{\xv}{\yv}^\perp}}}=\Alt{r-1}{\puv{\xv}{i}}{\puv{\yv}{i}(\puv{\xv}{i}-x_i)}^\perp.
  \end{equation}
  This means that starting from a generic alternant code $\Alt{r}{\xv}{\yv}$ of degree $r$, we can derive in polynomial time, by first computing two auxiliary codes $\AC =  \left(\sh{i}{\Alt{r}{\xv}{\yv}}\right)^\perp$ and 
  $\BC=\sqb{\sh{i}{\Alt{r}{\xv}{\yv}^\perp}}$, and then computing the conductor $\cond{\AC}{\BC}$ of $\AC$ into $\BC$, an alternant code of degree $r-1$. As we will see,  there are only two conditions to be met for performing this task:
  \begin{enumerate}
  \item[(i)] $r \geq q+1$ where $q$ is the alphabet size of the alternant
    code,
  \item[(ii)] $\sqb{\sh{i}{\Alt{r}{\xv}{\yv}^\perp}}$ is not the full code
    $\fq^{n-1}$ where $n$ is the codelength of the alternant code.
  \end{enumerate}
 By iterating this process, we can compute in polynomial time some kind of ``filtration'' of duals of alternant codes of decreasing degree 
 \begin{equation}
 \label{eq:filtration}
 \AC_r^\perp = \Alt{r}{\xv}{\yv}^\perp \sss{i_1} \AC_{r-1}^\perp  \sss{i_2} \cdots  \sss{i_{r-q}}{\AC_q}^\perp
 \end{equation}
  with multipliers and support which are related to the original support and multiplier (and from which the original support and multiplier can be easily recovered). Here the notation $\AC \sss{i} \BC$ means that 
  \[
  \sh{i}{\AC} \supseteq \BC.
  \]
  What can we do with this sequence? The point is that if the degree of the alternant code is small enough, we can compute its support and multiplier by solving a low degree algebraic system related to the algebraic systems considered in \cite{FOPT10,FGOPT13}. We will detail this in the particular case where $r=3$ and show that in this case, solving the system can be performed in polynomial time with Gr\"obner basis techniques. Roughly speaking the reason for this is that we have a conjunction of factors : a very overdetermined and highly structured system which gives during the Gr\"obner basis computation many new very low degree equations. We will also show that it is possible to speed up significantly the system solving process by introducing in the algebraic modeling new low degree polynomial equations which are not in the ideal of the original algebraic equations from \cite{FGOPT13} and which express the fact that the multiplier vector has only non zero entries and the support vector has only distinct entries. This will result in the end in a very efficient system solving procedure. Note that the aforementioned procedure reaches an alternant code $\AC_3$ of degree $3$  when the field size $q$ is either equal to $2$ or $3$. In other words, we have at the end a way to break a McEliece scheme based on {\em binary or ternary alternant codes}  as soon as $\sqb{\sh{i}{\Alt{r}{\xv}{\yv}^\perp}}$ is not the full code $\fq^{n-1}$. By using the formula given in \cite{MT21}, this is the case when
 \begin{eqnarray}
n-1 & > & \binom{rm+1}{2}-\frac{m}{2}(r-1)\left((2e+1)r-2\frac{q^{e}-1}{q-1}\right), \label{eq:un}\\
\text{where }e &\eqdef &\max\{i \in \mathbb{N} \mid r\ge q^i+1\}=\floor{\log_q(r-1)} \nonumber
 \end{eqnarray} 
We give in Table \ref{tab:summary} the known cases where it is possible to attack a McEliece scheme based on alternant codes together with the new attack proposed here:
\begin{table}[h!]
\begin{tabular}{|c|c|}
\hline 
paper & restriction \\ \hline 
\cite{SS92,CGGOT14} & $m=1$ \\ \hline
\cite{COT14} & $m=2$ + Wild Goppa code \\ \hline
this paper & $q=2$ or $q=3$, $m$ arbitrary + high rate condition \eqref{eq:un} \\ 
& (does not apply in the particular case of Goppa codes) \\ \hline
\end{tabular}
  \caption{Summary of polynomial time attacks on McEliece schemes based on alternant codes with the conditions to apply them.\label{tab:summary}}
\end{table}

In a nutshell, our contribution can be summarized as follows.

\paragraph{\em Breaking the $m=2$ barrier.}
 It is has been a long standing open problem after the \cite{SS92} attack on McEliece-GRS whether it is also possible to attack subfield subcodes of GRS codes, i.e. attack McEliece-alternant or McEliece-Goppa. A first step in this direction was made in \cite{CGGOT14} where a new attack on McEliece-GRS was derived with a hope to generalize it to McEliece-alternant or McEliece-Goppa because it is in essence only based on the fact that certain alternant or Goppa codes behave differently from random codes with respect to the dimension of the square code.  This was confirmed in \cite{COT14} by attacking McEliece-wild Goppa in the particular case where the extension degree $m$ is $2$, but the method used there which uses squares of the (shortened) Goppa code is bound to fail for higher extension degrees. Here we break for the first time the $m=2$ barrier, which was even conjectured at some point to be the ultimate limit for such algebraic attacks to work in polynomial time and show that we can actually attack McEliece-alternant for {\em any} extension degree $m$  provided that the rate of the alternant code is sufficiently large \eqref{eq:un} and the field size sufficiently low $q=2$ or $q=3$. Our attack is also based on square code considerations, but this time on the {\em dual} of the alternant code. The point is that in this case the square of the dual can also be distinguished from a random code in this regime \cite{FGOPT11,MP12}.  The attack is however much more involved in this case, because the dual loses the simple polynomial evaluation formulation of the Goppa code, since it is in this case the trace of a GRS code and not subfield subcode of a GRS code. Understanding the structure of the square is more complicated as was already apparent in \cite{MT21} which started such a task.
 
 \smallskip
 \paragraph{\em Opening the road for attacking the CFS scheme.}
 Interestingly our attack does not work at all when the alternant code has the additional structure of being a Goppa code. However this work could open the road for also attacking this subcase, in which case we could hope to break the CFS scheme \cite{CFS01} which operates precisely in the high rate regime where the square of the dual of the Goppa code behaves abnormally.

 \smallskip
\paragraph{\em New  algebraic attack.}
Our attack consists in two phases, the first phase computes a filtration of the dual of the alternant code  by computing iteratively conductors and the second phase solves with Gröbner bases techniques a variation of the algebraic system considered in \cite{FGOPT13} and recovers the support and the multiplier from the dual of the alternant code of degree $3$ we have at the end of the filtration when $q=2$ or $q=3$. We improve rather significantly upon the complexity of solving this system by adding new equations expressing the constraints on the support (all elements are distinct) and the multipliers (all elements are nonzero). By using certain heuristics that we confirmed experimentally we are able to prove that the Gröbner basis computation takes polynomial time and give a complete algebraic explanation of each step of the computation. It is likely that this analysis could be carried over for larger constant degree alternant codes. This would allow to break McEliece-alternant for larger field size than $3$.

A proof-of-concept implementation in MAGMA of the whole attack can be found at \url{https://github.com/roccomora/HighRateAlternant}.

 \section{Notation and prerequisites}
This section fixes the notation and reviews some tools and constructions needed for understanding the filtration attack and the rest of the article. These include the notion of subfield subcodes, trace codes, puncturing and shortening operators.

\subsection{Notation}

\subsubsection*{Finite fields.}
We denote by $\F$ a generic finite field and we will use it whenever it is not important to specify the field size. We will extensively consider the finite field extension $\Fq\subseteq \Fqm$, where $\Fq$ and $\Fqm$ are two finite fields with $q$ and $q^m$ elements respectively, for a prime power $q$ and a positive integer $m$.

\subsubsection*{Vector and matrix  notation.}
Vectors are indicated by lowercase bold letters $\xv$ and matrices by uppercase bold letters $\Mm$. By convention, vectors coordinates are indexed starting from 1. Moreover, $\Iintv{a}{b}$ indicates the closed integer interval between $a$ and $b$. Given a function $f$ acting on $\F$ and a vector $\xv=(x_i)_{1\le i \le n} \in \F$, the expression $f(\xv)$ is the component-wise mapping of $f$ on $\xv$, i.e. $f(\xv)=(f(x_i))_{1\le i \le n}$. We will even apply this with functions $f$ acting on $\F \times \F$: for instance for two vectors
$\xv$ and $\yv$ in $\F^n$ and two positive integers $a$ and $b$ we denote by $\xv^a\yv^b$ the vector $(x_i^a y_i^b)_{1 \leq i \leq n}$. It will also be convenient to use a notation to say that we drop a few positions in a vector. If $\xv=(x_i)_{i \in \Iintv{1}{n}}$ and $\Ical$ is a subset of positions, we denote by $\puv{\xv}{\Ical}$ the vector 
$\puv{\xv}{\Ical} \eqdef (x_i)_{i \in \Iintv{1}{n}\setminus \Ical}$.  When there is just one position $i$ in $\Ical$ we simply write
$\puv{\xv}{i}$ in this case (for instance, if $\xv=(x_1,x_2,x_3)$ then $\puv{\xv}{2}=(x_1,x_3)$).
\subsubsection*{Vector spaces}
Vector spaces are indicated by $\CC$. For two vector spaces $\CC$ and $\DC$, the notation $\CC\oplus \DC$ means that the two vector spaces are in direct sum, i.e.  that $\CC\cap\DC=\{0\}$.
\subsection{Trace codes and subfield subcodes}

A  useful $\Fq$-linear map from $\Fqm$ to its subfield $\Fq$ is the \textit{trace operator}:
\begin{definition}
Given the finite field extension $\Fqm/\Fq$, we define the \textit{field trace} $\textrm{Tr}_{\Fqm/\Fq}$ for all $x\in \Fqm$ as the $\Fq$-linear map from $\Fqm$ to $\Fq$ such that
\[
\Trm{x}=\sum_{i=0}^{m-1} x^{q^i}.
\]
The definition extends to vectors $\xv\in\Fqm^n$ so that the trace acts component-wise:
\[
\Trm{\xv}=(\Trm{x_1},\dots,\Trm{x_n})
\]
and hence to codes $\CC$ over $\Fqm$
\[                                                                                      
\Trm{\CC}=\{\Trm{\cv} \mid \cv \in \CC\}.
\]
\end{definition}
We will omit the subscript $\Fqm/\Fq$ and write more coincisely $\textrm{Tr}$ from now on, whenever the extension field is clear from the context.
Another standard construction to move from codes over $\Fqm$ to codes over a subfield $\Fq$ is the one of subfield subcode.
\begin{definition}
	Let $\CC\subseteq \Fqm^n$ be a code. The code
	\[
		\CC_{\rvert \Fq}\eqdef \CC \cap \Fq^n = \{\cv \in \CC \mid \forall i \in \Iintv{1}{n},\,c_i \in \Fq\}
	\]
	is called \textit{subfield subcode} over $\Fq$ of $\CC$.
\end{definition} 
It is easy to see that $\CC_{\rvert \Fq}$ is an $\Fq$-linear subspace of $\Fqm^n$, where $n$ is the length of $\CC$.
We conclude this subsection by stating the main classical result linking trace codes to subfield subcodes, i.e. Delsarte's theorem.
\begin{thm}{(Delsarte's Theorem \cite{D75})} \label{thm: Delsarte}
Let $\CC$ be a code over $\Fqm$. Then
\[
(\CC_{\rvert \Fq})^\perp=\Trm{\CC^\perp}.
\]
\end{thm}

\subsection{Shortening and Puncturing}
\begin{definition}
Given a code $\CC \subseteq \F^n$ and a subset $\Ical \subseteq \Iintv{1}{n}$, the \textit{punctured} code $\pu{\Ical}{\CC}$ and the \textit{shortened} code $\sh{\Ical}{\CC}$ over $\Ical$ are defined respectively as
\begin{align*}
\pu{\Ical}{\CC}=& \left\{ (c_i)_{i \in \Iintv{1}{n}\setminus \Ical } \mid \cv \in \CC\right\},\\
\sh{\Ical}{\CC}=& \left\{ (c_i)_{i \in \Iintv{1}{n}\setminus \Ical } \mid \exists\,\cv=(c_i)_{i \in \Iintv{1}{n}} \in \CC \text{ s.t. }\forall\,i \in \Ical, c_i = 0\right\}.
\end{align*}
\end{definition}
For the sake of simplicity, when $\Ical = \{i\}$, we denote the punctured and the shortened codes in $\Ical$ with $\pu{i}{\CC}$ and $\sh{i}{\CC}$ respectively.

Since we will frequently move from a code to its dual, it is worthwhile to mention how the shortening and puncturing combine with the dual operator:
\begin{prop} \cite[Theorem~1.5.7]{HP03} \label{prop: sh_pu_dual}
Let $\CC$ be a linear code of length $n$ and $\Ical \subset \Iintv{1}{n}$. Then
\[
	\sh{\Ical}{\CC^\perp}=\pu{\Ical}{\CC}^\perp \quad \text{and} \quad \pu{\Ical}{\CC^\perp}=\sh{\Ical}{\CC}^\perp.
\]
\end{prop}
Finally we recall the well known fact that a shortened alternant codes is itself an alternant code 
\begin{prop} \label{prop:sh_alt}
  Let $\Alt{r}{\xv}{\yv}$ be an alternant code of length $n$ and $\Ical\subseteq\{1,\dots, n\}$. Then
  \[
    \sh{\Ical}{\Alt{r}{\xv}{\yv}}=\Alt{r}{\puv{\xv}{\Ical}}{\puv{\yv}{\Ical}}.
  \]
\end{prop}

\subsection{Base field extension and alternant codes}
 
 It will be convenient to consider for a code defined over $\Fq$ its extension by scalars over $\Fqm$, meaning the following.
 \begin{definition}[extension of a code over a field extension]
 Let $\CC$ be a linear code over $\Fq$. We denote by $\CC_{\Fqm}$ the $\Fqm$-linear span of $\CC$ in $\Fqm^n$.
 \end{definition}
 It could seem that this is a somewhat trivial object, but this notion comes very handy in our case, since the extension to $\Fqm$ of the dual of an alternant code defined over $\Fq$ is a sum of $m$ GRS codes as shown by
  
  \begin{prop}\label{prop:dual_alt_fqm}
Let $\Alt{r}{\xv}{\yv}$ be an alternant code over $\Fq$. Then
\[
\left(\Alt{r}{\xv}{\yv}^\perp\right)_{\Fqm} = \sum_{j=0}^{m-1} \GRS{r}{\xv}{\yv}^{q^j}.
\]
where $\CC^{q^j} \eqdef \{\cv^{q^j}=(c_i^{q^j})_i: \cv \in \CC\}$ is readily seen to be an $\Fqm$-linear code of the same dimension as $\CC$ when $\CC$ is itself an $\Fqm$-linear code
\end{prop}
\begin{proof}
It follows directly from Proposition~\ref{thm: Delsarte} by taking $\CC\eqdef \GRS{r}{\xv}{\yv}$ and the
fact that we have for any nonnegative integer $j$ and any $\Fqm$-linear code $\CC$  (see Proposition \ref{prop:trace_fqm} in Appendix \ref{sec:extended}).
\[
\Tr{\CC}_{\Fqm} = \CC + \CC^q + \cdots + \CC^{q^{m-1}}.
\]
\end{proof}
We will also need the following straightforward properties
\begin{lemma}\label{lem:R15} \cite[Lemma~2.23]{R15}
	Let $\CC,\DC \subseteq \Fq^n$ be two codes. Then
	\begin{enumerate}
		\item \label{it:dual} $(\CC^\perp)_{\Fqm}=(\CC_{\Fqm})^\perp \subseteq \Fqm^n$.
		\item \label{it:inclusion}$\CC\subseteq \DC \iff \CC_{\Fqm}\subseteq \DC_{\Fqm}$.
		\item \label{it:sum} $(\CC+\DC)_{\Fqm} = \CC_{\Fqm}+\DC_{\Fqm}$ and $(\CC\oplus\DC)_{\Fqm} = \CC_{\Fqm}\oplus\DC_{\Fqm}$.
		\item \label{it:cap} $(\CC\cap\DC)_{\Fqm} = \CC_{\Fqm}\cap\DC_{\Fqm}$.
		\item \label{it:stp} $(\starp{\CC}{\DC})_{\Fqm} = \starp{\CC_{\Fqm}}{\DC_{\Fqm}}$.			
	\end{enumerate}
\end{lemma}

\subsection{A few useful notions for algebraic system solving}
We will use here techniques inspired by Gröbner bases computations and it will be useful to recall a few useful notions about this topic.
We will consider here ideals generated by polynomials over some field $\F$ in a certain number of variables, say $X_1,\dots,X_n$.
We will use two orders on the monomials of $\F[X_1,\cdots,X_n]$. The first is just the usual lexicographic order where the single variables 
are ordered as $X_1 > X_2 > \cdots > X_n$ and two monomials $X^\alpha \eqdef X_1^{\alpha_1}\cdots X_n^{\alpha_n}$ and $X^\beta \eqdef X_1^{\beta_1}\cdots X_n^{\alpha_n}$ are ordered as follows: $X^\alpha > X^\beta$ if and only if there exists $i \in \Iintv{1}{n}$ such that 
$\alpha_i > \beta_i$ and $\alpha_j=\beta_j$ for $j <i$. We will also use the {\em graded reverse lexicographical order (grevlex)} meaning that  $X^\alpha > X^\beta$ if and only if there exists $i \in \Iintv{1}{n}$ such that 
$\alpha_i < \beta_i$ and $\alpha_j=\beta_j$ for $j > i$.
We will also need the notion of {\em leading monomial}. The leading monomial $\LM(P)$ of a polynomial (for a given order on the monomials) is the largest monomial appearing in $P$. In other words for $P= \sum_{\alpha \in \Z^n} c_\alpha X^\alpha$ with $P \in \F[X_1,\cdots,X_n]$ where we use the notation $X^\alpha=X_1^{\alpha_1}\cdots X_n^{\alpha_n}$ we have
$\LM(P) = \max\{X^\alpha \mid c_\alpha \neq 0\}$ where the maximum is taken with respect to the order on the monomials which is chosen.

Another notion will be useful for manipulating an algebraic system, namely
\begin{definition}[Macalauy matrix]\label{def:macaulay}
The Macaulay $\macaulay{f_1,\cdots,f_k}$ matrix associated to an algebraic system (and with respect to a certain order of the monomials) specified by the polynomials $(f_1, \dots, f_k)$ is the $k \times M$ matrix where $M$ is the number of different monomials 
$m_1 > m_2 \cdots > m_M$ involved in the $f_i$'s where the entry $m_{ij}$ is the coefficient
of the monomial $m_j$ in $f_i$.
\end{definition}

 \section{The filtration}
\label{sec:filtration}
Assume that $r\ge q+1$.
We are going to explain in this  section  how, starting with the code $\AC_r = \Alt{r}{\xv}{\yv}$,  we are (generally) able to compute in polynomial time a sequence of alternant codes such that 
\begin{equation}
 \label{eq:filtration2}
 \AC_r^\perp = \Alt{r}{\xv}{\yv}^\perp \sss{i_1} \AC_{r-1}^\perp \sss{i_2} \cdots \sss{i_{r-q}}{\AC_q}^\perp,
 \end{equation}
  where all the alternant codes have a support which is easily derived from the support of $\AC_r$, since it just amounts to drop some positions. 
 This is instrumental for recovering efficiently the algebraic structure of the alternant code (i.e. the support $\xv$ and the multiplier $\yv$) in what follows. The core of this attack is the following theorem.
 \begin{thm}\label{thm:main_filtration}
Let $\Alt{r}{\xv}{\yv}$ be an alternant code such that $r\ge q+1$. Let $\CC\eqdef \left(\sh{i}{\Alt{r}{\xv}{\yv}}\right)^\perp, \DC\eqdef
\sqb{\sh{i}{\Alt{r}{\xv}{\yv}^\perp}}$, for an arbitrary position $i$. Then
\[
\starp{\Alt{r-1}{\puv{\xv}{i}}{\puv{\yv}{i}(\puv{\xv}{i}-x_i)}^\perp}{\CC}\subseteq \DC,
\]
or, equivalently,
\[
\cond{\CC}{\DC}\supseteq  \Alt{r-1}{\puv{\xv}{i}}{\puv{\yv}{i}(\puv{\xv}{i}-x_i)}^\perp.
\]
\end{thm}
Experimentally it turns out that we actually have equality here $\cond{\CC}{\DC} =  \Alt{r-1}{\puv{\xv}{i}}{\puv{\yv}{i}(\puv{\xv}{i}-x_i)}^\perp$ when choosing a random alternant code. It is however possible to build artificial examples where equality does not hold. Notably, we also found that the subfamily of Goppa codes does not meet this property either. However, if $\xv$ and $\yv$ are sampled at random, we never met a case in our experiments where equality does not hold. This leads us to state the following conjecture.
\begin{conjec} \label{heur:main_filtration}
Let $\Alt{r}{\xv}{\yv}$ be a random alternant code over $\Fq$, such that $r\ge q+1$ and $\sqb{\Alt{r}{\xv}{\yv}^\perp}$ is not the full code. 
Let $\CC\eqdef \left(\sh{i}{\Alt{r}{\xv}{\yv}}\right)^\perp$ and $\DC\eqdef\sqb{\sh{i}{\Alt{r}{\xv}{\yv}^\perp}}$, for an arbitrary position $i$. Then, we expect that
\[
\cond{\CC}{\DC} = \Alt{r-1}{\puv{\xv}{i}}{\puv{\yv}{i}(\puv{\xv}{i}-x_i)}^\perp.
\]
with  probability $1-o(1)$ as $n \to \infty$.
\end{conjec}

It is clear that this conjecture, if true, allows to compute in polynomial time the filtration \eqref{eq:filtration2}, since computing conductors just amounts to solve a linear system. Moreover, if $\xv,\yv$ are taken at random, then we expect a random behaviour for $\puv{\xv}{i}$ and $\puv{\yv}{i}(\puv{\xv}{i}-x_i)$ as well. The fact that when taking a random alternant code, the whole filtration can be computed has indeed been verified experimentally. The first conductor $\AC_{r-1} = \Alt{r-1}{\puv{\xv}{i_1}}{\puv{\yv}{i_1}(\puv{\xv}{i_1}-x_{i_1})}^\perp$ is computed by using directly Conjecture \ref{heur:main_filtration} and we iterate the process by choosing a sequence of positions $i_1$, $i_2,\cdots,i_{r_q}$ by which we shorten. We let
$\Ical_s =\{i_1,\cdots,i_s\}$. It is readily seen that we compute iteratively from
$\AC_{r-s+1}^\perp \eqdef \Alt{r-s+1}{\puv{\xv}{\Ical_{s-1}}}{\puv{\yv}{\Ical_{s-1}}\prod_{j=1}^{s-1}(\puv{\xv}{\Ical_{s-1}}-x_{i_j})}^\perp$ the code
\[
\AC_{r-s}^\perp \eqdef \Alt{r-s}{\puv{\xv}{\Ical_{s}}}{\puv{\yv}{\Ical_{s}}\prod_{j=1}^{s}(\puv{\xv}{\Ical_{s}}-x_{i_j})}^\perp.
\]
This allows to decrease the degree of the alternant one by one. The last step ends by using the conjecture with $r=q+1$ and ends with the conductor $\AC_q^\perp$. Let us now prove Theorem \ref{thm:main_filtration}.

\subsection{Proof of Theorem \ref{thm:main_filtration}}

It will be convenient to prove a slightly stronger result which implies Theorem \ref{thm:main_filtration}. It is based on the observation that
$\Alt{r-1}{\puv{\xv}{i}}{\puv{\yv}{i}(\puv{\xv}{i}-x_i)}^\perp\subseteq \sh{i}{\Alt{r}{\xv}{\yv}^\perp}$. Theorem~\ref{thm:main_filtration} is indeed implied by the following slightly stronger result:
\begin{thm}\label{thm:equivalent}
Let $\Alt{r}{\xv}{\yv}$ be an alternant code such that $r\ge q+1$. Let $\CC\eqdef \left(\sh{i}{\Alt{r}{\xv}{\yv}}\right)^\perp$ and $\DC'\eqdef\starp{\Alt{r-1}{\puv{\xv}{i}}{\puv{\yv}{i}(\puv{\xv}{i}-x_i)}^\perp}{\sh{i}{\Alt{r}{\xv}{\yv}^\perp}}$, for an arbitrary position $i$. Then
\[
\starp{\Alt{r-1}{\puv{\xv}{i}}{\puv{\yv}{i}(\puv{\xv}{i}-x_i)}^\perp}{\CC}\subseteq \DC',
\]
or, equivalently,
\[
\cond{\CC}{\DC'}\supseteq  \Alt{r-1}{\puv{\xv}{i}}{\puv{\yv}{i}(\puv{\xv}{i}-x_i)}^\perp.
\]
\end{thm}

To prove this theorem, it will be convenient to consider the extensions of all these codes to $\Fqm$, in other words we are going to prove that
\begin{equation}\label{eq:equivalent}
	\starp{\BC_{\Fqm}}{\CC_{\Fqm}}\subseteq \DC'_{\Fqm},
\end{equation}
where $\BC \eqdef \Alt{r-1}{\puv{\xv}{i}}{\puv{\yv}{i}(\puv{\xv}{i}-x_i)}^\perp$. The point of doing this, is that  (i) it is equivalent to prove \eqref{eq:equivalent} because of the points \ref{it:inclusion} and \ref{it:stp} of Lemma \ref{lem:R15}, (ii) the extended codes can be expressed as a sum of GRS codes due to Proposition \ref{prop:dual_alt_fqm}.
The proof of Theorem \ref{thm:equivalent} will proceed by following the steps below.
\begin{itemize}
\item[{\bf Step 1:}]We first observe that the code $\CC_{\Fqm} = \left( \left(\sh{i}{\Alt{r}{\xv}{\yv}}\right)^\perp\right)_{\Fqm}
=\left( \Alt{r}{\puv{\xv}{i}}{\puv{\yv}{i}}^\perp \right)_{\Fqm}$ (where the last equality follows from 
Proposition \ref{prop:sh_alt}) decomposes as
\[
\left( \Alt{r}{\puv{\xv}{i}}{\puv{\yv}{i}}^\perp \right)_{\Fqm} = \left( \sh{i}{\Alt{r}{\xv}{\yv}^\perp} \right)_{\Fqm} \oplus \Fqmspan{\puv{\yv}{i}}.
\]
This is Lemma \ref{lem:decomposition} below. This implies that 
\[
\starp{\BC_{\Fqm}}{\CC_{\Fqm}} = \underbrace{\starp{\BC_{\Fqm}}{\left( \sh{i}{\Alt{r}{\xv}{\yv}^\perp} \right)_{\Fqm}}}_{\DC'_{\Fqm}} +\starp{\BC_{\Fqm}}{\Fqmspan{\puv{\yv}{i}}}.
\]
Therefore in order to prove \eqref{eq:equivalent} it will be enough to prove the inclusion 
\begin{equation}
\label{eq:step1} \starp{\BC_{\Fqm}}{\Fqmspan{\puv{\yv}{i}}} \subseteq  \DC'_{\Fqm}.
\end{equation}
\item[{\bf Step 2:}] To achieve this purpose, we then prove that the extended shortened code $\left( \sh{i}{\Alt{r}{\xv}{\yv}^\perp}\right)_{\Fqm}$ contains
as a subcode $ \BC_{\Fqm}=\left(\Alt{r-1}{\puv{\xv}{i}}{(\puv{\xv}{i}-x_i)\puv{\yv}{i}}^\perp\right)_{\Fqm}$ (Lemma \eqref{lem:subcodes}) on one hand and
$\CC_0  \eqdef  \Fqmspan{y_i^{q^u}\puv{\yv}{i}^{q^v}-y_i^{q^v}\puv{\yv}{i}^{q^u}:\;u,v \in \Iintv{0}{m-1} }$ on the other hand. Actually more is true, namely that the extended shortened code is a sum of these two subcodes but we will not need this.
\item[{\bf Step 3:}] By using this, in order to prove \eqref{eq:step1} we will start with an element in $ \starp{\BC_{\Fqm}}{\Fqmspan{\puv{\yv}{i}}}$ and by adding suitable elements of $\starp{\BC_{\Fqm}}{\CC_0}$ and $\sq{\BC_{\Fqm}}$ we will show that we end with an element in $\DC'_{\Fqm}=\starp{\BC_{\Fqm}}{\left(\sh{i}{\Alt{r}{\xv}{\yv}^\perp}\right)_{\Fqm}}$.
\end{itemize}
Let us now state and prove the lemmas we have mentioned above.
\begin{lemma}\label{lem:decomposition} 
\[
 \left( \Alt{r}{\puv{\xv}{i}}{\puv{\yv}{i}}^\perp \right)_{\Fqm} = \left( \sh{i}{\Alt{r}{\xv}{\yv}^\perp} \right)_{\Fqm} \oplus \Fqmspan{\puv{\yv}{i}}.
\]
\end{lemma}
\begin{proof}
Note that 
$\left(\Alt{r}{\xv}{\yv}^\perp\right)_{\Fqm}$ decomposes as the set of codewords $\AC_0$ that are zero in $i$ (this
is $\left( \sh{i}{\Alt{r}{\xv}{\yv}^\perp} \right)_{\Fqm}$ where add an extra-position at $i$ which is always 0) plus a 
space of dimension $1$ generated by an element of $\left(\Alt{r}{\xv}{\yv}^\perp\right)_{\Fqm}$ which is not equal to $0$ at position 
$i$. $\yv$ is clearly such an element and we can write
\[
\left(\Alt{r}{\xv}{\yv}^\perp\right)_{\Fqm} = \AC_0 \oplus \Fqmspan{\yv}.
\]
By puncturing these codes at $i$ we get our lemma.
\end{proof}
\begin{lemma}\label{lem:subcodes}
We have for any position $i$
\begin{eqnarray}
\left( \sh{i}{\Alt{r}{\xv}{\yv}^\perp}\right)_{\Fqm}  & \supseteq & \left(\Alt{r-1}{\puv{\xv}{i}}{(\puv{\xv}{i}-x_i)\puv{\yv}{i}}^\perp\right)_{\Fqm}  \label{eq:altsubcode}\\
 \left( \sh{i}{\Alt{r}{\xv}{\yv}^\perp}\right)_{\Fqm}  & \supseteq &\CC_0\;\;\text{where} \label{eq:Czerosubcode}\\
\CC_0 & \eqdef & \Fqmspan{y_i^{q^u}\puv{\yv}{i}^{q^v}-y_i^{q^v}\puv{\yv}{i}^{q^u}\mid\;u,v \in \Iintv{0}{m-1} }.\label{eq:defCzero}
\end{eqnarray}
\end{lemma}
\begin{proof}
By using Proposition \ref{prop:dual_alt_fqm} we know that
\begin{eqnarray*}
\left( \Alt{r}{\xv}{\yv}^\perp\right)_{\Fqm}  &= & \Fqmspan{\xv^{aq^\ell}\yv^{q^\ell}; \;a \in \Iintv{0}{r-1},\;\ell \in \Iintv{0}{m-1}},\\
\left(\Alt{r-1}{\puv{\xv}{i}}{(\puv{\xv}{i}-x_i)\puv{\yv}{i}}^\perp\right)_{\Fqm}   &= & \Fqmspan{{\puv{\xv}{i}}^{aq^\ell}(\puv{\xv}{i} - x_i)^{q^\ell}\puv{\yv}{i}^{q^\ell}; \;a \in \Iintv{0}{r-2},\;\ell \in \Iintv{0}{m-1}}.
\end{eqnarray*}
Observe now that 
\[\left( \sh{i}{\Alt{r}{\xv}{\yv}^\perp}\right)_{\Fqm}= \sh{i}{\left(\Alt{r}{\xv}{\yv}^\perp\right)_{\Fqm}}.\]
Clearly $\xv^{aq^\ell}(\xv - x_i)^{q^\ell}\yv^{q^\ell}=(\xv^a(\xv-x_i))^{q^\ell}\yv^{q^\ell}$ vanishes at $i$ and belongs to 
$\left( \Alt{r}{\xv}{\yv}^\perp\right)_{\Fqm}$ for $a$ in $\Iintv{0}{r-2}$. Therefore 
${\puv{\xv}{i}}^{aq^\ell}(\puv{\xv}{i} - x_i)^{q^\ell}\puv{\yv}{i}^{q^\ell}$ belongs to $\sh{i}{\left(\Alt{r}{\xv}{\yv}^\perp\right)_{\Fqm}}$. This proves 
\eqref{eq:altsubcode}. 
Similarly $y_i^{q^u}\yv^{q^v}-y_i^{q^v}\yv^{q^u}$ belongs clearly to $\left( \Alt{r}{\xv}{\yv}^\perp\right)_{\Fqm}$ and vanishes at $i$. Hence 
$y_i^{q^u}\puv{\yv}{i}^{q^v}-y_i^{q^v}\puv{\yv}{i}^{q^u}$ belongs to $\sh{i}{\left(\Alt{r}{\xv}{\yv}^\perp\right)_{\Fqm}}$. This proves 
\eqref{eq:defCzero}.
\end{proof}
\begin{lemma}\label{lem:BC}
Let $\BC  \eqdef \Alt{r-1}{\puv{\xv}{i}}{\puv{\yv}{i}(\puv{\xv}{i}-x_i)}^\perp$ and $\DC' \eqdef \starp{\BC}{\sh{i}{\Alt{r}{\xv}{\yv}^\perp}}$, then 
if $r \geq q+1$:
\[
\starp{\BC_{\Fqm}}{\Fqmspan{\puv{\yv}{i}}} \subseteq \DC'.
\]
\end{lemma}
\begin{proof}
We use the same notation as in Lemma \ref{lem:subcodes} and observe that $\DC_0$ and $\DC_1$ defined below are both subcodes of 
$\DC'_{\Fqm}$:
\begin{eqnarray*}
\DC_0 \eqdef \starp{\BC_{\Fqm}}{\CC_0} &= &  \Fqmspan{{\puv{\xv}{i}}^{aq^\ell}(\puv{\xv}{i} - x_i)^{q^\ell}
\left(y_i^{q^u}\puv{\yv}{i}^{q^v+q^\ell}-y_i^{q^v}\puv{\yv}{i}^{q^u+q^\ell}\right)\mid \;a \in \Iintv{0}{r-2},\ell,u,v \in \Iintv{0}{m-1}}\\
\DC_1 \eqdef \sq{\BC_{\Fqm}} &= & \Fqmspan{{\puv{\xv}{i}}^{aq^j+bq^\ell}(\puv{\xv}{i} - x_i)^{q^j+ q^\ell}\puv{\yv}{i}^{q^j+q^\ell}\mid \;a,b \in \Iintv{0}{r-2},\;j,\ell \in \Iintv{0}{m-1}}
\end{eqnarray*}
This is a direct consequence of $\DC' _{\Fqm} = \starp{\BC_{\Fqm}}{\left( \sh{i}{\Alt{r}{\xv}{\yv}^\perp}\right)_{\Fqm}}$ 
(definition of $\DC'$ and Point \ref{it:stp} in Lemma \ref{lem:R15})
and Lemma \ref{lem:subcodes}.
We also observe that
\[
	\starp{\BC_{\Fqm}}{\Fqmspan{\puv{\yv}{i}}}=\Fqmspan{\left(\puv{\xv}{i}\right)^{a q^\ell}(\puv{\xv}{i}-x_i)^{q^\ell}\left(\puv{\yv}{i}\right)^{q^\ell+1} \mid a \in \Iintv{0}{r-2}, l \in \Iintv{0}{m-1}}.
\]
Our proof strategy is to start with an element appearing in the vector span above and by suitable additions of elements of $\DC_i$ (where $i \in \{0,1\}$), and possibly also by multiplying by some elements in $\Fqm$, results at the end in an element in $\DC_i$. This will prove our lemma. We will 
use here the notation 
\[
\uv \stackrel{\DC_i}{\rightarrow} \vv
\]
to write that $\vv$ can be obtained from $\uv$ by adding a suitable element of $\DC_i$ and multiplying by 
some element in $\Fqm$, i.e. this is equivalent to $\uv-\lambda \vv \in \DC_i$ for a suitable element $\lambda$ in $\Fqm$.
It is readily seen that for any $a$ in $\Iintv{0}{r-2}$, $u$ and $v$ in  $\Iintv{0}{m-1}$ and any polynomial $P$ in $\Fqm[X]$ of degree $\leq r-2$ we have
\begin{eqnarray}
P(\puv{\xv}{i})^{q^\ell}(\puv{\xv}{i}-x_i)^{q^\ell} \puv{\yv}{i}^{q^\ell+q^u} & \stackrel{\DC_0}{\rightarrow} &
P(\puv{\xv}{i})^{q^\ell}(\puv{\xv}{i}-x_i)^{q^\ell} \puv{\yv}{i}^{q^\ell+q^v} \label{eq:reductionD0}\\
P(\puv{\xv}{i})^{q^\ell}(\puv{\xv}{i}-x_i)^{q^\ell} \puv{\yv}{i}^{2q^\ell}  & \stackrel{\DC_1}{\rightarrow} &
(\puv{\xv}{i}-x_i)^{q^\ell} \puv{\yv}{i}^{2q^\ell}. \label{eq:reductionD1}
\end{eqnarray}
The first reduction follows by noticing that
\begin{eqnarray*}
P(\puv{\xv}{i})^{q^\ell}(\puv{\xv}{i}-x_i)^{q^\ell} \puv{\yv}{i}^{q^\ell+q^u} &= & 
P(\puv{\xv}{i})^{q^\ell}(\puv{\xv}{i}-x_i)^{q^\ell} \puv{\yv}{i}^{q^\ell} y_i^{-q^v} \left( y_i^{q^v}\puv{\yv}{i}^{q^u}-
y_i^{q^u}\puv{\yv}{i}^{q^v} + y_i^{q^u}\puv{\yv}{i}^{q^v}\right)\\
& = & \dv + y_i^{q^u-q^v} P(\puv{\xv}{i})^{q^\ell}(\puv{\xv}{i}-x_i)^{q^\ell} \puv{\yv}{i}^{q^\ell+q^v}
\end{eqnarray*}
where $\dv= P(\puv{\xv}{i})^{q^\ell}(\puv{\xv}{i}-x_i)^{q^\ell} \puv{\yv}{i}^{q^\ell} y_i^{-q^v} \left( y_i^{q^v}\puv{\yv}{i}^{q^u}-
y_i^{q^u}\puv{\yv}{i}^{q^v} \right)$ clearly belongs to $\DC_0$.
The second reduction follows by performing the Euclidean division of $P(X)$ by $(X-x_i)$. We can namely write
$P(X)=(X-x_i)Q(X) + P(x_i)$ for a polynomial $Q$ of degree $\deg P-1$. Therefore
\begin{eqnarray}
P(\puv{\xv}{i})^{q^\ell}(\puv{\xv}{i}-x_i)^{q^\ell} \puv{\yv}{i}^{2q^\ell}  &= & \left( (\puv{\xv}{i}-x_i)Q(\puv{\xv}{i}) + P(x_i) \right)^{q^\ell} (\puv{\xv}{i}-x_i)^{q^\ell} \puv{\yv}{i}^{2q^\ell} \nonumber \\
& = & \left( (\puv{\xv}{i}-x_i)^{q^\ell} Q(\puv{\xv}{i})^{q^\ell} + P(x_i)^{q^\ell} \right)(\puv{\xv}{i}-x_i)^{q^\ell} \puv{\yv}{i}^{2q^\ell}
\label{eq:fqmlinearity} \\
& = & \dv + P(x_i)^{q^\ell}(\puv{\xv}{i}-x_i)^{q^\ell} \puv{\yv}{i}^{2q^\ell} \nonumber
\end{eqnarray}
where \eqref{eq:fqmlinearity} follows from the $\Fq$-linearity of the Frobenius action $x \mapsto x^{q^\ell}$ and $\dv=\left( (\puv{\xv}{i}-x_i)^{q^\ell} Q(\puv{\xv}{i})^{q^\ell}  \right)(\puv{\xv}{i}-x_i)^{q^\ell} \puv{\yv}{i}^{2q^\ell}
$ belongs obviously 
to $\DC_1$.

Let us show the inclusion by performing for a generator $\puv{\xv}{i}^{a q^\ell}(\puv{\xv}{i}-x_i)^{q^\ell}\puv{\yv}{i}^{q^\ell+1}$ of $\starp{\DC_{\Fqm}}{\Fqmspan{\puv{\yv}{i}}}$ a sequence of reductions
\begin{eqnarray*}
\puv{\xv}{i}^{a q^\ell}(\puv{\xv}{i}-x_i)^{q^\ell}\puv{\yv}{i}^{q^\ell+1}& \stackrel{\DC_0}{\rightarrow} &
\puv{\xv}{i}^{a q^\ell}(\puv{\xv}{i}-x_i)^{q^\ell}\puv{\yv}{i}^{2q^\ell} \\
& \stackrel{\DC_1}{\rightarrow} &
(\puv{\xv}{i}-x_i)^{q^\ell}\puv{\yv}{i}^{2q^\ell}
\end{eqnarray*}
The crucial argument is now the simple observation that 
\begin{eqnarray*}
(\puv{\xv}{i}-x_i)^{q^\ell} &= & \left( (\puv{\xv}{i}-x_i)^{q}\right)^{q^{\ell^-}}
\end{eqnarray*}
where $\ell^- \eqdef \ell-1$ if $\ell >0$ and $\ell^- \eqdef m-1$ if $\ell=0$. This is a consequence of the fact that the entries
of $\xv$ are in $\Fqm$. 
This suggests the following sequence of reductions
\begin{eqnarray*}
(\puv{\xv}{i}-x_i)^{q^\ell}\puv{\yv}{i}^{2q^\ell} & \stackrel{\DC_0}{\rightarrow} &
\left( (\puv{\xv}{i}-x_i)^{q}\right)^{q^{\ell^-}}\puv{\yv}{i}^{q^\ell + q^{\ell^-} }= \left( (\puv{\xv}{i}-x_i)^{q-1}\right)^{q^{\ell^-}}
(\puv{\xv}{i}-x_i)^{q^{\ell^-}} \puv{\yv}{i}^{ q^{\ell^-}+q^\ell }\\
& \stackrel{\DC_0}{\rightarrow} & \left( (\puv{\xv}{i}-x_i)^{q-1}\right)^{q^{\ell^-}} (\puv{\xv}{i}-x_i)^{q^{\ell^-}} \puv{\yv}{i}^{ 2q^{\ell^-} }=
\left( (\puv{\xv}{i}-x_i)^{q-2}\right)^{q^{\ell^-}} (\puv{\xv}{i}-x_i)^{2q^{\ell^-}} \puv{\yv}{i}^{ 2q^{\ell^-} }.
\end{eqnarray*}
Note that the last reduction could be performed because the degree of the polynomial $(\puv{\xv}{i}-x_i)^{q-1}$, which is $q-1$, is less than or equal to $r-2$ by assumption on $r$. For the very same reason ($r \geq q+1$) we observe that the right-hand term 
$\left( (\puv{\xv}{i}-x_i)^{q-2}\right)^{q^{\ell^-}} (\puv{\xv}{i}-x_i)^{2q^{\ell^-}} \puv{\yv}{i}^{ 2q^{\ell^-} }$ belongs to $\DC_1$ which 
finishes the proof.
\end{proof}

We are ready now to prove Theorem \ref{thm:equivalent}.

\begin{proof}[Proof of Theorem \ref{thm:equivalent}.]
From Lemma \ref{lem:decomposition}, we know that $\CC_{\Fqm}= \left( \Alt{r}{\puv{\xv}{i}}{\puv{\yv}{i}}^\perp \right)_{\Fqm} $ can be decomposed as
\[
\CC_{\Fqm} = \left( \sh{i}{\Alt{r}{\xv}{\yv}^\perp} \right)_{\Fqm} \oplus \Fqmspan{\puv{\yv}{i}}.
\]
This implies that 
\begin{eqnarray*}
\starp{\BC_{\Fqm}}{\CC_{\Fqm}} &= &\underbrace{\starp{\BC_{\Fqm}}{\left( \sh{i}{\Alt{r}{\xv}{\yv}^\perp} \right)_{\Fqm}}}_{\DC'_{\Fqm}} +\starp{\BC_{\Fqm}}{\Fqmspan{\puv{\yv}{i}}} \\
& = & \starp{\BC_{\Fqm}}{\left( \sh{i}{\Alt{r}{\xv}{\yv}^\perp} \right)_{\Fqm}} \;\;\text{(by Lemma \ref{lem:BC})}\\
& = & \DC'_{\Fqm}.
\end{eqnarray*}
The equality of the extended codes over $\Fqm$ implies the equalities of the codes over $\Fq$ which ends the proof.
\end{proof}

\subsection{What is wrong with Goppa codes?}\label{ss:Goppa}
Before moving to the second part of the attack, we make a short digression on how the arguments explained so far (do not) apply to the Goppa case. The discussion below does not represent a proof that computing a filtration is impossible for Goppa codes, but rather an intuition about what hampers it. Goppa codes behave differently from random alternant codes and provide counterexamples to Heuristic~\ref{heur:main_filtration}. The latter should be replaced by 

\begin{heur} \label{heur:Goppa_cond}
Let $\Goppa{\xv}{\Gamma} \eqdef\Alt{r}{\xv}{\yv}$ 
be a random Goppa code of degree $r$, with $r\ge q-1$ and $\sqb{\Goppa{\xv}{\Gamma}^\perp}$ being different from the full code. 
Choose an arbitrary code position $i$ and let $\CC\eqdef 
\left(\sh{i}{\Goppa{\xv}{\Gamma}}\right)^\perp$ and $\DC\eqdef\sq{\sh{i}{\Alt{r}{\xv}{\yv}^\perp}}$. Then, with high probability, 
\[
\cond{\CC}{\DC} = \Alt{r}{\puv{\xv}{i}}{\puv{\yv}{i}(\puv{\xv}{i}-x_i)}^\perp.
\]
\end{heur}
This is unfortunate, since our approach heavily builds upon the fact that the degree of the conductor decreases. Moreover, it will turn out that the code we obtain as a conductor could have been obtained directly by shortening a suitable code. Actually, it will turn out that for Goppa codes, there are several codes which are very close to each other and which are obtained by various shortenings. This is summarized by the following proposition, whose proof will be given in the appendix. We will explain in what follows why this 
phenomenon is the main obstacle for applying our conductor approach.

\begin{restatable}{prop}{propCodes}
Let $\Goppa{\xv}{\Gamma} \eqdef\Alt{r}{\xv}{\yv}$ be a Goppa code of degree $r$. We have for any code positions $i$ and $j$
with $i \neq j$:
\begin{eqnarray}
\Alt{r+1}{\xv}{\yv}^\perp & = & \Goppa{\xv}{\Gamma} + \Fqspan{\mathbf{1}} \;\;\text{(from \cite[prop. 1]{B00})} \label{eq:altrpu}\\
\Alt{r}{\puv{\xv}{i}}{\puv{\yv}{i}(\puv{\xv}{i}-x_i)}^\perp& = & \sh{i}{\Alt{r+1}{\xv}{\yv}^\perp} \label{eq:shi}\\
\sh{j}{\Alt{r}{\puv{\xv}{i}}{\puv{\yv}{i}(\puv{\xv}{i}-x_i)}^\perp} & = & \sh{i}{\Alt{r}{\puv{\xv}{j}}{\puv{\yv}{j}(\puv{\xv}{j}-x_j)}^\perp}.\label{eq:shicomshj}
\end{eqnarray}
\end{restatable}
	We can summarize these relationships with the diagram below, where arrows mean an inclusion of the lower code into the upper code and two arrows pointing at same code represent the intersection. The typical code dimensions are shown too.
	
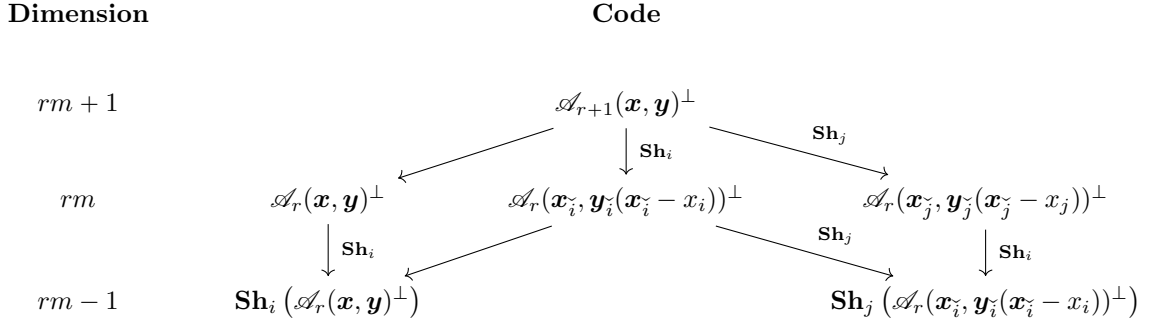
\[
		\hspace{-0.5cm}
		\begin{tikzcd}
			\textbf{Dimension} & & \textbf{Code} & \\
			rm+1 & & \Alt{r+1}{\xv}{\yv}^\perp \arrow[dl] \arrow[d, "\text{ \bf Sh}_i"] \arrow[dr, "\text{ \bf Sh}_j"]  & \\
			rm & \Alt{r}{\xv}{\yv}^\perp \arrow[d, "\text{ \bf Sh}_i"] & \Alt{r}{\puv{\xv}{i}}{\puv{\yv}{i}(\puv{\xv}{i}-x_{i})}^\perp \arrow[dl] \arrow[dr, "\text{ \bf Sh}_j"] &  \Alt{r}{\puv{\xv}{j}}{\puv{\yv}{j}(\puv{\xv}{j}-x_{j})}^\perp \arrow[d, "\text{ \bf Sh}_i"]\\
			rm-1 & \sh{i}{\Alt{r}{\xv}{\yv}^\perp} & & \sh{j}{\Alt{r}{\puv{\xv}{i}}{\puv{\yv}{i}(\puv{\xv}{i}-x_{i})}^\perp} 
		\end{tikzcd}
	\]

 \section{Algebraic cryptanalysis}
\label{sec:algebraiccryptanalysis}
The previous section explains how to obtain, under some conditions, the alternant code $\Alt{3}{\xv'}{\yv'}$ with support $\xv'=\puv{\xv}{_\Ical}$ and multiplier $\yv'=\puv{\yv}{\Ical} \left(\prod_{i \in \Ical} (\puv{\xv}{\Ical}-x_i)\right)$ for some $\Ical\subseteq \Iintv{1}{n}$ such that $\card{\Ical}=r-3$, and with length $n'=n-r+3$ and degree $3$, starting from the knowledge of the length-$n$ public code $\Alt{r}{\xv}{\yv}$. For the sake of clarity, in this section we perform algebraic cryptanalysis on the alternant code $\Alt{3}{\xv}{\yv}$ of length $n$. Essentially, we can ignore the structure of $\yv'$ and the decreased length because the filtration preserves the support and multiplier randomness and the code distinguishability. In Section~\ref{sec:interlacing}, we will see how to get back a support and a multiplier defining $\Alt{r}{\xv}{\yv}$ (not necessarily $\xv$ and $\yv$) from a support and a multiplier defining $\Alt{3}{\xv'}{\yv'}$ (not necessarily $\xv'$ and $\yv'$). Moreover, we will focus on the case $r=3$ for the system resolution, but the algebraic modeling we describe in Sections~\ref{sec:algebraicmodeling} and \ref{sec:reducingsolutions} is more general and makes sense for any $r\ge 3$.
The full attack needs the filtration to reach degree 3, and therefore works specifically for $q=2$ or $q=3$.  However, the specific modeling for $r=3$ described in  Section~\ref{sec:r=3} is valid for any field size. We describe in Sections~\ref{sec:qodd} and \ref{sec:qeven}  a polynomial time attack on alternant codes of degree 3 for any $q>2$.
This result is original, and to the best of our knowledge, no polynomial time attack was known on non-structured alternant or Goppa codes even for $r=3$.
Section~\ref{sec:alt:justif} contains the theoretical proofs and experimental validation for our algorithm.

For $q=2$, we also have a polynomial time attack on random alternant
codes of degree 3. However, the algorithm and the theoretical
justifications are quite different from the $q>2$ case, they involve
algebraic properties of binary alternant codes and are beyond the
scope of this article. We just sketch the solving algorithm in
Section~\ref{sec:q=2}.

\subsection{The algebraic modeling from \cite{FGOPT13}, $r\ge 3$}
\label{sec:algebraicmodeling}
We give an alternative definition of alternant codes which is more suitable for this section. To this end we first need to introduce a rectangular Vandermonde-like matrix:
\[\Vm_r(\xv,\yv)\eqdef
\begin{bmatrix}
y_1 & \dots & y_n \\
y_1 x_1 & \dots & y_n x_n \\
\vdots & \ddots & \vdots \\
y_1 x_1^{r-1} & \dots & y_n x_n^{r-1}
\end{bmatrix},
\]
where 
$\xv=(x_1,\dots,x_n)\in \Fqm^n$ and $\yv=(y_1\dots,y_n)\in \Fqm^n$.
It is readily seen that an alternant code $\Alt{r}{\xv}{\yv}$ can be expressed as the code with parity-check matrix $\Vm_r(\xv, \yv)$:
$$
\Alt{r}{\xv}{\yv}\eqdef\left\{\cv \in \Fq^n \mid \Vm_r(\xv, \yv)\cv^T=\zerom\right\}
$$

We will adopt the notation and follow the description of the algebraic model presented in \cite{FGOPT13}. We denote with $\Gm=(g_{i,j})\in\Fq^{k\times n}$ the $k\times n$ generator matrix of $\Alt{r}{\xv}{\yv}$.
The definition of alternant codes given above implies that 
\[
\Vm_r(\xv, \yv)\Gm^T=\zerom_{k\times n}.
\]
This matrix equation translates into the following system of polynomial equations:
\[
\left\{\sum_{j=1}^n g_{i,j}Y_j X_j^e =0 \mid i \in \Iintv{1}{k}, e \in \Iintv{0}{r-1}\right\},
\]
where $\Xm \eqdef \{X_1,\dots,X_n\}$ and $\Ym\eqdef\{Y_1,\dots,Y_n\}$ are two blocks of $n$ unknowns, each corresponding to the support and multiplier coordinates respectively. Observe that the sought vectors $\xv$ and $\yv$ satisfy indeed the polynomial system.

We can assume, up to a permutation of columns, that $\Gm$ is in systematic form, i.e. $\Gm=(\Im_k \mid \Pm)$, where $\Im_k$ is the identity matrix of size $k$ and $\Pm=(p_{i,j})$ for $i \in \Iintv{1}{k}, j \in \Iintv{k+1}{n}$. The polynomial system can be therefore rewritten as
 \begin{equation} \label{eq: systematic_form}
\left\{Y_i X_i^e =- \sum_{j=k+1}^n p_{i,j}Y_j X_j^e \mid i \in \Iintv{1}{k}, e \in \Iintv{0}{r-1}\right\}.
\end{equation}

As explained in \cite{FGOPT13}, thanks to the systematic form assumption, we can get rid of several variables and consider an algebraic system in only $2(n-k)$ unknowns. Indeed it is possible to exploit the trivial identity for $i \in \Iintv{1}{k}$
$$
Y_i  (Y_i X_i^2) = (Y_i X_i)^2
$$
and replace the term $Y_iX_i^a$ for $a \in \Iintv{0}{2}$ by 
$-\sum_{j=k+1}^n p_{i,j}Y_jX_j^a$
to obtain the equation $\left(\sum_{j=k+1}^n p_{i,j}Y_j\right) \left(\sum_{j=k+1}^n p_{i,j}Y_jX_j^2\right)=\left( \sum_{j=k+1}^n p_{i,j}Y_jX_j\right)^2$.
After putting all terms on one side and expanding the products, we obtain that the polynomial
\[
f_i \eqdef \sum_{j, j'\in\Iintv{k+1}{ n} }p_{i,j}p_{i,j'} Y_j Y_{j'} (X_{j'}^2-X_j X_{j'})
\]
must be zero on the support $\xv$ and multiplier $\yv$. 
 Since 
\begin{eqnarray*}
f_i
& =&  \sum_{k+1\le j< j'\le n}
p_{i,j}p_{i,j'} Y_j Y_{j'} (X_j^2 + X_{j'}^2-2X_j X_{j'})\\
& =&  \sum_{k+1\le j< j'\le n}
p_{i,j}p_{i,j'} Y_j Y_{j'} (X_j-X_{j'})^2
\end{eqnarray*}
we finally obtain the following algebraic system 
\begin{align}
  \label{eq:S}
  \cS \eqdef \left\{ \sum_{k+1 \leq j < j' \leq n} p_{i,j}p_{i,j'} Y_j Y_{j'} (X_{j}-X_{j'})^2\mid i \in \Iintv{1}{k} \right\}.
\end{align}

\subsection{Reducing the number of solutions}
\label{sec:reducingsolutions}
The ideal generated by $\cS$ is not zero-dimensional. It will be
convenient here to reduce to this case by specializing appropriately
some variables.
The positive dimension of this ideal is due in the first instance to the degrees of freedom for the support and multiplier coordinates. 
In essence this is due to the fact that a homography $z \mapsto \frac{a z + b}{cz + d}$ maps the support $\xv$ of an alternant code to another support describing the same alternant code (but possibly with a different multiplier) at the condition that $c x_i + d$ never vanishes. When there exists a value $x_i$ of the support of the alternant code for which $c z_i + d=0$, the resulting code is not an alternant code, but belongs to a slightly larger family of codes : it will be a subfield subcode of a Cauchy code. 
Let us recall its definition taken from \cite{D87}.
Given a field $\F$ we can identify the projective line $\mathcal{P}^1(\F)$ with $\bar{\F}\eqdef\F \cup \{\infty\}$, where the symbol $\infty$ is called \textit{point at infinity}, through the map $\phi: \bar{\F}\rightarrow \F^2 \setminus\{0\}, \phi(e)=(e , 1)$ if $e \in \F$ and $\phi(\infty)=(1,0)$. Moreover, let $\F[W,Z]_{l}^H$ be the set of homogeneous polynomials of degree $l$ in two variables $W, Z$. Given $P \in \F[W,Z]_{l}^H$ and $e \in \bar{\F}$, we define $P(e)\eqdef P(\phi(e))$. Then
\begin{definition}[Cauchy code]
	Let $\xv\eqdef(x_1,\dots,x_n) \in \bar{\F}^n$ be a vector of distinct elements and $\yv \eqdef (y_1,\dots,y_n) \in \F^n$  be a vector of nonzero elements. Let $r \in \Iintv{0}{n}$. The \textbf{Cauchy code} $\CC_r(\xv, \yv)$ is defined as
  \[
  \CC_r(\xv,\yv)\eqdef\{(y_1 P(x_1),\dots,y_n P(x_n)) \mid P \in \F[W,Z]_{r-1}^H\}.
  \]
  As in the case of generalized Reed-Solomon codes, $\xv$ is called a {\em support}  and $\yv$ a {\em multiplier} of the Cauchy code.
\end{definition}
If we assume $x_n=\infty$, a generator matrix of $\CC_r(\xv, \yv)$ is given by
\begin{equation} \label{eq: gen_cauchy}
\begin{bmatrix}
y_1 & \dots & y_{n-1} & 0 \\
y_1 x_1 & \dots & y_{n-1} x_{n-1} & 0 \\
\vdots & \ddots & \vdots \\
y_1 x_1^{r-1} & \dots & y_{n-1} x_{n-1}^{r-1} & y_n
\end{bmatrix}.
\end{equation}
On the other hand, when $\xv \in \F^n$, \ie when all the $x_i$'s are different from $\infty$, the Cauchy code $\CC_r(\xv,\yv)$ can be easily seen as $\GRS{r}{\xv}{\yv}$. They are also MDS codes. In some sense, Cauchy code are the projective linearization of GRS codes. Analogously subfield subcodes of Cauchy codes generalize subfield subcodes of GRS codes, \ie alternant codes. 

One of the main result of \cite{D87} was to characterize the possible supports and multipliers of Cauchy codes. In particular, it is proven there that
\begin{thm}\cite{D87} \label{thm: cauchy}
Let $r  \in \Iintv{2}{n-2}$. Then $\CC_r(\xv, \yv) = \CC_r(\xv', \yv')$ if and only if there exists a homography $f(z)= \frac{az +b}{cz+d}$ ($a,b,c,d \in \F$, $ad-bc \neq 0$) such that 
$\xv' = f(\xv)$ and $\yv' = \lambda  \theta(\xv)^{r-1} \yv$ where  $\lambda \in \F\setminus\{0\}$ and 
\begin{align*}
\theta(z)= cz+d &\quad \text{if } z \in \F \text{ and } cz+d\neq 0,\\
\theta(z)= (ad-bc)/(-c) &\quad \text{if } z \in \F \text{ and } cz+d = 0,\\
\theta(\infty)= c &\quad \text{if } c\neq 0,\\
\theta(\infty)= a &\quad \text{if } c= 0.
\end{align*}
\end{thm}
Since the elements $a,b,c,d$ in $f(z)=\frac{az+b}{cz+d}$ are defined up to a multiplication by a nonzero scalar, Theorem~\ref{thm: cauchy} pragmatically implies that we are allowed to fix three variables in block $\Xm$ and one in block $\Ym$. The price to pay for this additional specialization is that now we have to eventually handle the point at infinity. We have seen that the column corresponding to the point at infinity in the generator matrix of a Cauchy code has a special form and this changes for this coordinate the form of the system $\cS$ given in \eqref{eq:S}. The problem is that we do not know a priori which $x_i$ will be infinite. To circumvent this problem, we choose the value $x_i$ that will  be set to infinity, say $x_n$.

It is also convenient to make the other choices to belong to the subfield $\Fq$ over which the alternant code is defined: all the Gröbner bases computations will stay in the subfield and this results in slightly improved computation times. Finally we make the following choice (which also simplifies slightly the analysis of the Gröbner basis computations):

\begin{align}\label{eq:specialization}
X_{n-2}=0,\quad X_{n-1}=1,\quad X_n=\infty,\quad Y_n=1,
\end{align}
From now on, we denote with a prime the systems where the previous values have been specialized. For instance $\cS'$ corresponds to $\cS$ with $X_{n-2}=0, X_{n-1}=1, X_n=\infty, Y_n=1$.
For the particular specialization $X_n=\infty$, the last column of the Vandermonde matrix is $\trsp{
  \begin{pmatrix}
    0 & \dots & 0 & 1
  \end{pmatrix}
}$ (see~\eqref{eq: gen_cauchy}), and the shape of $\cS'$ is a little bit different of the one of $\cS$ (see next section).

However, the set of solutions of the system $\cS'$ still contains a component of positive dimension $n-k-1=rm-1$, that corresponds to the solutions of $\{Y_j=0 \mid k+1\le j \le n-1\}$. The classical way to deal with the parasite solutions $Y_j=0$ is to introduce to the system a new variable $T_j$ together with the polynomial $T_jY_j-1$. This ensures that $Y_j=0$ is not a solution to the system. However, this also adds variables to the system, and increases the degree of the polynomials during a Gröbner basis computation. A similar phenomenon occurs with the constraints $X_j-X_{j'}\ne 0$. 
We solve these problems in an easier way in Steps~\eqref{step:2} and \eqref{step:3} of the algorithm presented in Section~\ref{sec:qodd}.

If we substitute \eqref{eq:specialization} into $\cS$ for $r\ge 4$, we
just get the same system as if we had specialized $Y_{n}=0$. In the $r=3$ case it is slightly different.
\subsection{Specialized system $\cS'$ in the $r=3$ case}
\label{sec:r=3}
\begin{proposition}
For $r=3$, we can choose part of the support and multiplier as $X_{n-2}=0$, $X_{n-1}=1$, $X_n=\infty$, $Y_n=1$ and obtain the following algebraic system
\begin{align}
\label{eq:Sprime}
  \cS' &\eqdef &\left\{
                 \sum_{k+1\le j< j'\le n-1}
                 p_{i,j}p_{i,j'} Y_j Y_{j'} (X_{j}-X_{j'})^2   +  \sum_{j=k+1}^{n-1} p_{i,j}p_{i,n}Y_j
                 \mid i \in \Iintv{1}{k} \right\}.
\end{align} 
\end{proposition}
\begin{proof}
The fact that we can choose the support and the multiplier in this way follows from the fact that homographies act 3-transitively on the projective plane and from Theorem \ref{thm: cauchy}. To obtain the algebraic system we proceed similarly to what was done to obtain the algebraic system $\cS$: this time the Vandermonde matrix is
\[
  \Vm_3(\Xm, \Ym)\eqdef 
  \left[
    \begin{array}{rrr|rrr}
      Y_1 & \dots & Y_{n-3} & Y_{n-2} & Y_{n-1} & 0 \\
      Y_1 X_1 & \dots & Y_{n-3} X_{n-3} & 0 & Y_{n-1} & 0 \\
      Y_1 X_1^2 & \dots & Y_{n-3} X_{n-3}^2 & 0 & Y_{n-1} & 1 \\
    \end{array}
  \right],
\]
so that we get the relations
\begin{eqnarray*}
Y_i & = &                 \sum_{k+1\le j\le n-1} p_{i,j} Y_j\\
Y_i X_i & = & \sum_{k+1\le j\le n-1} p_{i,j} Y_j X_j \\
Y_i X_i^2 &= & \sum_{k+1\le j\le n-1} p_{i,j} Y_j X_j^2  + p_{i,n}.
\end{eqnarray*}
Substituting them into the relations
$
Y_i  (Y_i X_i^2) = (Y_i X_i)^2
$ for $i \in \Iintv{1}{k}$ gives $\cS'$.
\end{proof}

Our purpose is to solve the algebraic system $\cS'$ that is given by
$k$ polynomials of bidegree $(2,2)$ in the two blocks of variables $\Xm$
and $\Ym$ and involves $2(n-k)$ variables.  The polynomials can be
expressed in terms of the $\binom{n-k}{2}$ variables
\begin{align}
Z_{j,j'} = 
                 \begin{cases}
                   Y_j Y_{j'} (X_{j}-X_{j'})^2 & \text{ for } k+1\le j < j' \le n-1,\\
                   Y_j & \text{ for } j'=n, \text{ as } X_{n}=\infty \text{ and } Y_n=1.
                 \end{cases}
\end{align}
The classical way to solve such system is to first try to linearized
the system using the variables $Z_{j,j'}$, for that purpose it is
important to know the rank of the associated linearized system.

Its rank is trivially upper bounded by the number of expressions $Z_{j,j'} $, i.e. by $\binom{n-k}{2}$. However, in the high rate regime, \cite{FGOPT13} proposed an algebraic heuristic explaining why there is a tighter upper bound. This heuristic was confirmed by intensive computations which showed that the upper-bound was indeed tight. In \cite{MT21} a proof of the upper bound was given. These results lead us to make the following assumptions, given for the case $r=3$ of interest, that will be satisfied for a random alternant code.

\begin{assumption}[Random alternant code]\label{ass:dimA}
  We assume that $\code{A}_3(\xv,\yv)$ is in standard form, and that
  its dimension satisfies $k = n-rm = n-3m$.
\end{assumption}
\begin{assumption}[High rate regime] \label{assumption: rank}
If  $q\ge 3$, we assume that
\begin{align*}
\rank(\cS) =\rank(\cS') = \binom{3m}{2}-m \le k.
\end{align*}
\end{assumption}
Note that for $q=2$, the rank is smaller:
\begin{assumption}[High rate regime] \label{assumption: rank q=2}
If  $q=2$, we assume  that
\begin{align*}
  \rank(\cS) = \rank(\cS') = \binom{3m}{2}-3m \le k.
\end{align*}
\end{assumption}
This implies that, even after the change of variables $Z_{j,j'}$, the number of unknowns is larger than the number of independent polynomials, and linearization techniques are not enough to solve the system. Therefore, in the following we are going to explain how to tackle this problem with more advanced techniques, namely Gr\"obner basis. 

Note that for any solution $(\xv,\yv)$ to the system, the vector $(\xv^q,\yv^q)$ defined by taking the $q$-st power coordinate for each coordinate is still a solution to the system. This mean that we can expect at least $m$ solutions.

\subsection{The algorithm for $q$ odd}
\label{sec:qodd}
Generic Gr\"obner basis algorithms are not expected to solve efficiently systems with the same degree and same number of unknowns and equations as the one described before. Here however, some expedients can be taken into account to exploit the very strong algebraic structure and specific shape of the equations involved.

We present here an ad hoc algorithm based on Gr\"obner basis techniques, that
recovers the secret key in polynomial time.
We start from the system $\cS'$ given in \eqref{eq:Sprime}, under
Assumptions~\ref{ass:dimA} and \ref{assumption: rank}, and we
agglomerate in our strategy the constraints $Y_j\ne 0$ and
$X_j-X_{j'}\ne 0$. This section deals with the odd case (\ie when $q$
is the power of an odd prime), and next one with the even case $q=2^s$
with $s>1$.
This covers for instance the case $q=3$, for which a full key recovery
can be achieved, thanks to the filtration of alternant codes. The
proofs and detailed explanations are postponed in
Section~\ref{sec:alt:justif}.
Our algorithm consists of the following steps:
\begin{enumerate}
\item \label{step:1} {\bf Echelonizing system $\cS'$} We compute a basis of the $\Fq$-vector space $\code{S'}$ generated
  by $\cS'$. It has the shape
  \begin{align}\label{eq:1stbasis}
    \begin{cases}
      Y_jY_{j'}(X_j-X_{j'})^2 + L_{j,j'}(Y_{\ell}: \ell\in I)& \forall k+1\le j < j' \le n-1\\
      Y_{n-2}Y_{n-1} + L_{n-2,n-1}(Y_{\ell}:\ell\in I) \\
      Y_j - L_{j,n}(Y_{\ell}: \ell \in I)  & \forall k+1\le j \le n-1, j \notin I.
    \end{cases}
   \end{align}
   where $I\subset\Iintv{k+1}{n-1}$ has size $m$, and the $L_{j,j'}$'s are
   linear functions of the $Y_\ell$'s, $\ell\in I$.  In particular, it
   contains $2m-1$ homogeneous linear polynomials in $Y$, that come
   from the choice $X_n=\infty$ and $Y_n=1$ (see
   Proposition~\ref{prop:rankS} for the proof). This can be done in
   $\bigO{m^{2\omega}}$ operations in $\Fq$, where $\omega$ is the
   exponent of linear algebra.
\item \label{step:2} {\bf Removing the $Y_j = 0$ component} For each
   $j\in\Iintv{k+1}{n-1}$, we prove that there exists a set of $2m-1$
   linearly independent polynomials in $\code{S'}$ that are multiple
   of $Y_j$. As we know that our solution satisfies $Y_j\ne 0$, we add
   to the system the set $\cV_j$ of these polynomials divided by $Y_j$
   (see Proposition~\ref{prop:Vj} for details and proof). This has the
   effect to add $(2m-1)(3m-1)$ linearly independent
   polynomials of degree $3$ to the system, and to remove the
   non-zero-dimensional component from the solution set.  Note that
   Step~\eqref{step:1} corresponds to the computation of $\cV_n$, as
   $Y_n=1$. The cost for all $j$ is $\bigO{m^{\omega+1}}$.
\item \label{step:3} {\bf Adding bilinear polynomials} For
  each $j\in \Iintv{k+1}{n-1}\setminus\lbrace n-2\rbrace$, we consider
  the vector spaces $\code{U}_{j,n-2}$ formed by the polynomials
  $p$ such that $X_jp\in \code{V}_j+\code{V}_{n-2}$, where
  $\code{V}_j$ is the $\Fq$-vector space generated by $\cV_j$. We
  prove in Proposition~\ref{prop:Uj} that
  $\dim_{\Fq}(\code{U}_{j,n-2})\ge m$. Experimentally, this set has
  dimension exactly $m$. As we know that our solution satisfies
  $x_j\ne 0=x_{n-2}$, we add to the system a basis $\cU_{j,n-2}$ of
  $\code{U}_{j,n-2}$. This has the effect to add  (at least) $m(3m-2)$ linearly
  independent polynomials of degree $2$ to the system. The cost for all $j$ is
  $\bigO{m^{\omega+1}}$.
\item \label{step:4} {\bf Eliminating $2m-1$ variables $Y_j$ using the
  linear polynomials from Step~\eqref{step:1}} We now eliminate $2m-1$
  variables $Y$ from the polynomials in $\cU_{j,n-2}$ using
  the $2m-1$ homogeneous linear polynomials in $Y$ from
  Step~\eqref{step:1}. This step is heuristic, and verified experimentally:
  the resulting system admits a basis with the shape:
  \begin{align*}
    \begin{cases}
    Y_j X_{j'} + f_{j,j'}(Y_\ell : \ell \in I, 1), &\text{ for all } j\in I, j'\in\Iintv{k+1}{ n-3},\\
    \sum_{j\in I} a_jY_j - 1, & \text{ with } a_j\in\fq.
  \end{cases}
  \end{align*}
  where the $f_{j,j'}$ are linear functions of the $Y_\ell$'s. The
  basis contains an affine linear polynomial in $Y$, so
  that we get in total all $2m$ linearly independent linear
  polynomials in $Y_j$ that we can expect (see
  Proposition~\ref{prop:dimalternant}). The cost of this step is $\bigO{m^{2\omega}}$.
\item \label{step:5} {\bf Computing linear polynomials in the $X$
  variables} Provided that Step~\eqref{step:4} occurred
  as described, we deduce for each $j$ one affine linear polynomial
  expressing $X_j$ in terms of the $Y$'s
  (see Proposition~\ref{prop:linX}):
  \begin{align*}
\left\{    X_{j'} + \sum_{j\in I} a_jf_{j,j'}(Y_\ell : \ell \in I, 1) \mid j'\in\Iintv{k+1}{n-3}\right\}.
  \end{align*}
  The cost is $\bigO{m^2}$.
\item {\bf Computing the Gröbner basis} \label{step:6} By eliminating the
  $X_i$'s from the polynomials $Y_jX_{j'}+f_{j,j'}(Y_\ell : \ell\in I, 1)$, we get the final
  grevlex Gröbner basis of the system, that is experimentally
  \begin{align*}
    \begin{cases}
    Y_jY_{j'} - L'_{j,j'}(Y_i: i \in I_1,1) & j, j'\in I_1\\
    X_j - L'_{X_j}(Y_i: i \in I_1,1) & j \in \Iintv{k+1}{n-3}\\
    Y_j - L'_{Y_j}(Y_i: i \in I_1,1) & j \notin I_1
  \end{cases}
  \end{align*}
  for a set $I_1\subset \Iintv{k+1}{n-1} $ of size $m-1$, where
  the functions $L'$ are affine linear function.  This describes a
  variety of dimension 0 with $m$ solutions, that are exactly the $m$
  solutions obtained by applying the Frobenius morphism. The cost is
  $\bigO{m^{2\omega}}$.
\item {\bf Computing the solutions} \label{step:7} The lex basis can be
  obtained using the FGLM Algorithm from \cite{FGLM93} with $\bigO{m^{4}}$
  operations in $\Fq$, and allows to retrieve the $m$ solutions by
  factorization of a polynomial over $\Fq$ of degree $\le m$.
\item \label{step:8} The final step consists in retrieving separately the values
  for  $Y_1,\dots,Y_k$ and $X_1,\dots,X_k$  from \eqref{eq: systematic_form} for $e=0$ and $e=1$. This
  costs $\bigO{nm}$ operations in $\Fq$.  This needs to be done only once, since with the chosen specialization any of the $m$ solutions is a valid pair of support and multiplier coordinates, thus we can choose arbitrarily any of them.
\end{enumerate}
Note that all steps are just linear algebra, and the total complexity
is polynomial in $m$ and $n$, the global cost being $\bigO{m^{2\omega}+nm}$ as $m,n\to\infty$.

\subsection{The algorithm for $q$ even, $q>2$}
\label{sec:qeven}
We present here the differences for the case of characteristic 2. We
still assume Assumptions~\ref{ass:dimA} and \ref{assumption:
  rank}. Note that the case $q=2$ is very different, we will consider
this case later in Section~\ref{sec:q=2}. We assume here that $q=2^s$
with $s>1$.

First of all, in this case the system $\cS'$ can be rewritten as
\[
  \cS'=\left\{\sum_{k+1\le j<j'\le n-1} p_{i,j}p_{i,j'}Y_j Y_{j'}
    (X_j^2+X_{j'}^2) + \sum_{j=k+1}^{n-1} p_{i,j}p_{i,n}Y_j \mid
    i\in\Iintv{1}{k}\right\}.
\]
Since the $X_j$'s variables appear in the system with power 2 only, we can perform a change of variables by defining $W_j \eqdef X_j^2$, so that the system becomes
\begin{align}\label{eq:Sprime2}
\cS'_{2}=\left\{\sum_{k+1\le j<j'\le n-1} p_{i,j}p_{i,j'}Y_j Y_{j'} (W_j+W_{j'}) + \sum_{j=k+1}^{n-1} p_{i,j}p_{i,n}Y_j \mid i\in\Iintv{1}{k}\right\}.
\end{align}
Therefore, polynomials in $\cS'_2$ have bidegree $(1,2)$ in
$\Wm\eqdef\{W_1,\dots,W_n\}$ and $\Ym$ respectively. This simple trick
provides an effective speed up to the resolution.

Steps \eqref{step:1}--\eqref{step:2} are identical, we compute for
each $j\in\Iintv{k+1}{n-1}$ the vector space $\code{V}_j$ that has
dimension $2m-1$. The difference is that the polynomials in $\cV_j$
for $j\le n-1$ are linear combination of the monomials
$Y_{j'}(W_j+W_{j'})$, $j'\in\Iintv{k+1}{n-1}$, $j'\ne j$, and 1, hence
already bilinear. Step~\eqref{step:3} is unnecessary.

Step~\eqref{step:4}: using the linear
polynomials in $\cV_n$, we can directly eliminate $2m-1$ variables
$Y_j$ from the $(3m-1)(2m-1)$ bilinear polynomials in
$\cup_{j=k+1}^{n-1}\cV_j$. We obtain experimentally a basis of the
shape
  \begin{align}
  \begin{cases}
    Y_jW_{j'} + g_{j,j'}(Y_\ell: \ell \in I, 1),
    &\text{ for all } j\in I, j'\in\Iintv{k+1}{ n-3},\\
    \sum_{j\in I} a_jY_j - 1, & \text{ with } a_j\in\fq.
  \end{cases}
  \end{align}

  Steps \eqref{step:5}--\eqref{step:8} are identical to the odd case.
Remark that retrieving the
  value $x_j$ from the values ${w}_j$ is just a square root
  computation in characteristic 2.

  The final asymptotic cost is the same, $\bigO{m^{2\omega}+nm}$ as $n,m\to\infty$.

\subsection{Theoretical  and experimental validation of the algebraic algorithm}
\label{sec:alt:justif}
We start with a property that will be useful to determine the number of
linearly independent linear polynomials in $\code{S}'$, for any $q$ and any $r\ge 3$. Recall that we assumed without loss of generality that $y_{n}=1$.
\begin{proposition}
  \label{prop:dimalternant}
  Let  $\CC$ be the $\Fqm$ linear code  generated by $(y_{k+1},\dots,y_{n-1})$ in $\Fqm^{mr-1}$. Then, under Assumption~\ref{ass:dimA}, we have
  \begin{align*}
    \dim_{\Fq}(\trace{\CC}) &= m,\\
    \dim_{\Fq}(\dual{\trace{\CC}})&=(r-1)m-1.
  \end{align*}
  As any $\Fq$-linear combination of the $y_j$'s that is equal to zero
  provides a codeword in $\dual{\trace{\CC}}$, therefore there cannot
  be more than $(r-1)m-1$ linearly
  independent homogeneous $\Fq$-linear polynomials in
  $Y_{k+1},\dots,Y_{n-1}$ which cancel on $y_{k+1},\dots,y_{n-1}$, and no more than $(r-1)m$ linearly independent affine $\Fq$-linear
  polynomials in $Y_{k+1},\dots,Y_{n-1}$ that cancel on
  $(y_{k+1},\dots,y_{n-1},1)$.

  Equivalently, for all $j\in\Iintv{ k+1}{n-1}$, the code $\CC_j'\subset\Fqm^{mr-1}$ generated by
  $
  (y_{k+1}(x_{k+1}-x_{j})^{r-1},\dots,y_{n-1}(x_{n-1}-x_{j})^{r-1},1)\in\Fqm^{mr}$ and punctured in position $j-k$
  satisfies 
  \begin{align*}
    \dim_{\Fq}(\trace{\CC_j'}) &= m,\\
    \dim_{\Fq}(\dual{\trace{\CC_j'}})&=(r-1)m-1.
  \end{align*}
\end{proposition}
\begin{proof}
  If the code $\Alt{r}{\xv}{\yv}$ has dimension $k=n-mr$ and is in
  standard form, then the last $n-k=mr$ columns of its parity-check
  matrix must be an information set, i.e. the last $mr$ columns of
  $\Vm_r(\xv,\yv)$ must generate a trace code with dimension $mr$.
  This means in particular that the first row 
  $(y_{k+1},\dots,y_{n-1},0)$  must have rank weight $m$, and this is the
  same for all $r$ rows.  Then, by elementary combination of rows, the
  trace code of the following code must still have dimension $m$:
  \begin{align*}
    \begin{pmatrix}
            0 & y_{k+2}(x_{k+2}-x_{k+1})^{r-1} & \dots  & y_{n-1}(x_{n-1}-x_{k+1})^{r-1} & 1
    \end{pmatrix}
  \end{align*}
  and this can be done for any $x_j$ instead of $x_{k+1}$, hence the
  proposition.
\end{proof}
We now give details for our algorithm for $r=3, q\ge 3$.
\subsubsection*{Step~\eqref{step:1}: Echelonizing system $\cS'$}
 We linearize the set of polynomials \eqref{eq:Sprime}, by replacing
 ``polynomials'' by variables, instead of classically replacing any
 monomial by a new variable. The variables we consider are:
 \begin{align*}
   Z_{j,j'} & =
              \begin{cases}
 Y_jY_{j'}(X_j-X_{j'})^2 & \text{ for } k+1\le j < j' \le n-1\\
 Y_j  & \text{ for } j\in\Iintv{ k+1}{ n-1}, j'=n.
\end{cases}
 \end{align*}

 \begin{proposition}\label{prop:rankS}
   Assume that $q\ge 3$.
   Under Assumptions~\ref{ass:dimA} and \ref{assumption: rank}, for
   any set $I\subset\Iintv{k+1}{n-1}$ of size $m$ such that
   $\dim_{\fq}(\langle y_\ell : \ell \in I\rangle_{\fq})=m$, a basis
   of $\code{S'}$ is given by
   \begin{align*}
    \begin{cases}
      Y_jY_{j'}(X_j-X_{j'})^2 + L_{j,j'}(Y_{\ell}: \ell\in I)& j,j' \in \Iintv{k+1}{n-1},\;j < j'\\
      Y_{n-2}Y_{n-1} + L_{n-2,n-1}(Y_{\ell}:\ell\in I) \\
      Y_j - L_{j,n}(Y_{\ell}: \ell \in I)  & \forall j \in \Iintv{k+1}{n-1} \setminus I.
    \end{cases}
   \end{align*}
   where the $L_{j,j'}$ are linear functions of the $Y_\ell$'s,
   $\ell\in I$ (note that $L_{j,j'}$ implicitly depends on $I$). This basis
   can be computed in time $\bigO{m^{2\omega}}$ where $\omega$ is the
   constant of linear algebra.

   For $q$ even, we get the same basis for $\cS'_2$ with the terms
   $(X_j-X_{j'})^2$ replaced by $W_j+W_{j'}$.
\end{proposition}
 \begin{proof}
   We have $\binom{3m}{2}-m$ polynomials in $\binom{3m}{2}$
  variables. Among the variables, $3m-1$ are of degree 1 (the $Y_j$'s
  for $k+1\le j \le n-1$), one is of degree 2
  ($Z_{n-2,n-1} = -Y_{n-2}Y_{n-1}$, as $X_{n-2}=0$ and $X_{n-1}=1$)
  and the last $\binom{3m}{2}-3m$ are of degree 4. We can eliminate
  from the system all terms of degree $4$ and $2$. As the
  polynomials are linearly independent, we get at least $2m-1$ linear
  polynomials in the $Y_j$'s.

  As by Proposition~\ref{prop:dimalternant}, we have at most
  $2m-1$ linear relations between the $Y_{k+1},\dots,Y_{n-1}$,
  hence we have exactly $2m-1$ linear polynomials in the $Y_j$'s
  expressing any $Y_j$ in terms of the $\{Y_{\ell} : \ell \in I\}$ for
  some $I\subset\Iintv{k+1}{n-1}$ of size $m$, and all other
  polynomials express the terms of degree $\ge 2$ in terms of the
  $\{Y_{\ell} : \ell \in I\}$.

  To compute the basis it is enough to compute an echelon form of a
  matrix of size $(\binom{3m}{2}-m)\times\binom{3m}{2}$, the cost is
  $\bigO{m^{2\omega}}$.
 \end{proof}

 \subsubsection*{Step \eqref{step:2}: removing the $Y_j=0$ component}
The linear polynomials we get come from the fact that we have
specialized the $n$-th component to $x_n=\infty$ and $y_n=1$. Here we
show that it is equivalently possible to introduce the constraint
$Y_j\ne 0$ for all $j\in\Iintv{ k+1}{n-1}$. Let
$\Xm=(X_{k+1},\dots,X_{n-3})$ and $\Ym=(Y_{k+1},\dots,Y_{n-1})$. We
define the vector spaces
\begin{align}
  \code{V}_j &= \frac1{Y_j}\left(\code{S'}\cap Y_j\cdot \fq{}[\Xm,\Ym]_{\le 3}\right), & j\in\Iintv{k+1}{n-1}
\end{align}
that is
$\code{V}_j \eqdef \langle \frac{h_1}{Y_j},\dots,\frac{h_\ell}{Y_j}\rangle_{\fq}$
where $\lbrace h_1,\dots,h_\ell\rbrace$ is a basis of $\code{S'}\cap (Y_j\cdot
\fq{}[\Xm,\Ym]_{\le 3})$. We also define
\begin{align}
    \code{V}_n &= \langle Y_j - L_{j,n}(Y_\ell, \ell \in I)\rangle_{j\in\lbrace k+1..n-1\rbrace, j\notin I}.
\end{align}
\begin{proposition}\label{prop:Vj}
Under Assumptions~\eqref{ass:dimA} and \eqref{assumption: rank}, for any $q\ge 3$,
  \begin{align}
    \dim_{\fq}(\code{V}_j) &= 2m-1 & \forall j \in \Iintv{k+1}{n}\end{align}
  and any polynomial in $\code{V}_j$ is a linear combination of the $3m-1$
terms
  \begin{align*}
    \begin{cases}
      Y_{j'}(X_j-X_{j'})^2, &  j'\in\Iintv{k+1}{n-1}, j'\ne j\\
      1
    \end{cases}
  \end{align*}
  We also have, for each $j_1,j_2\in \Iintv{k+1}{ n-1}$ with $j_1\ne j_2$:
  \begin{align*}
    \dim_{\fq}(\code{V}_{j_1}+\code{V}_{j_2}) &= 4m-2.
  \end{align*}
  A basis $\cV_j$ of $\code{V}_j$ can be computed in time $\bigO{m^\omega}$ from the basis~\eqref{eq:1stbasis} of $\code{S'}$ and the set of all
  $\cV_j$'s can be computed in time $\bigO{m^{\omega+1}}$.
\end{proposition}
\begin{proof}
  Choose any set $I\subset\Iintv{k+1}{n-1}$ of size $m$ such
  that  $\dim_{\fq}\left(\Fqmspan{( y_{\ell})_{\ell \in I}}\right)=m$.
  
  a) Consider first the case where $j\notin I$, and $j\le n-3$.  To
  compute $\code{V}_j$, we just take the polynomials in~\eqref{eq:1stbasis}
  that contain $Y_j$, they are the $3m-1$ polynomials:
  \begin{align*}
    \begin{cases}
      Y_jY_{j'}(X_j-X_{j'})^2 + L_{j,j'}(Y_{\ell}: \ell\in I)& \text{ for } j'\in\Iintv{ k+1}{ n-1}, j'\ne j\\
      Y_j - L_{j,n}(Y_{\ell}: \ell \in I).
    \end{cases}
  \end{align*}
We have $3m-1$ linearly independent polynomials in $m$ variables
  $Y_{\ell}:\ell\in I$ and $3m-1$ variables $ Y_jY_{j'}(X_j-X_{j'})^2$
  and $Y_j$. By eliminating the $Y_{\ell}:\ell\in I$ we get at least
  $2m-1$ polynomials that are multiple of $Y_j$.

  b) If $j\in I$, $j\le n-3$, we have
  $\dim_{\fq}\left( \Fqmspan{ (y_i)_{i\in\Iintv{k+1}{n-1}\setminus \{j\}}}\right)
  \in\lbrace m-1,m\rbrace$. If the dimension over $\Fq$ is $m$, then
we can take a different set $I$ that generates a $\Fq$-vector space
$\Fqmspan{(y_i)_{i\in I}}$ of dimension $m$ such that $j\notin I$. If
the dimension over $\Fq$ is $m-1$, then for any
$j'\in\Iintv{k+1}{n-1}\setminus I$, the linear polynomials
$Y_{j'}-L_{j',n}(Y_\ell : \ell \in I)$ does not involve $Y_j$ (or we
would have a linear polynomial expressing $Y_j$ as a $\Fq$-linear
combination of the terms $Y_{j'}$ and
$Y_{\ell} : \ell \in I\setminus\lbrace j\rbrace$, which is impossible
considering that
$\textstyle{\dim_{\fq} \left(\Fqmspan{ (y_i)_{i\in\Iintv{k+1}{n-1}}}\right)=m}$). In
this case, we take the $3m-2$ polynomials involving
$Y_jY_{j'}(X_j-X_{j'})^2$ for all $j'\ne j$, they contains those
$3m-2$ terms, the variable $Y_j$ and $m-1$ variables
$Y_\ell : \ell \in I, \ell \ne j$. By eliminating the
$Y_\ell, \ell \in I, \ell \ne j$ we get at least $2m-1$ linear
polynomials multiple of $Y_j$.

  c) For $j\in\Iintv{ n-2}{n-1}$, it is exactly the same as a) and b), but with
  one more polynomial involving one more variable $Y_{n-2}Y_{n-1}$.

  In any case, after dividing by $Y_j$, we get at least $2m-1$
  polynomials involving the monomials $Y_{j'}(X_{j'}-X_{j})^2$ for
  $j'\in\Iintv{k+1}{ n-1}\setminus\{ j\}$ and $1$, and all those
  polynomials evaluate to zero on the support and multiplier of the
  code.  Now, according to Proposition~\ref{prop:dimalternant} with
  $r=3$, $\dim_{\Fq}(\trace{\code{C}_j'})=m$, where $\code{C}_j'$ is the code of length $3m-1$
  generated by
  $ ( y_{k+1}(x_{k+1}-x_{j})^{2},\dots,y_{n-1}(x_{n-1}-x_{j})^{2},1)$
  and punctured in position $j-k$ (where its value is 0). Hence there
  can be at most $2m-1$ linear polynomials between the $3m-1$ terms
  $1$ and $Y_{j'}(X_{j'}-X_j)^2$, $j'\in\Iintv{k+1}{n-1}, j'\ne j$.

  Finally, the polynomials in two different $\code{V}_j$'s are
  linearly independent, as the only common polynomial that could
  belong simultaneously to two different $\code{V}_j$'s is $1$, and the ideal is not
  generated by $1$.
\end{proof}
 \subsubsection*{Step \eqref{step:3}: adding bilinear polynomials}
 The system given by the union of $\cS'$ and the $(3m-1)(2m-1)$ cubic polynomials $\cV_j, j \in\Iintv{k+1}{n-1}$ determined at Step \eqref{step:2} generates a zero-dimensional ideal, whose variety contains exactly $m$ solutions.
 It would be enough to run a Gr\"obner basis for this new system, in order to retrieve the support and multiplier. However, specifically for the $q$ odd case, we are able to deepen the analysis and use efficiently the constraints about support coordinates, \ie $X_{j_1}- X_{j_2}\ne 0$, by computing efficiently a set of bilinear polynomials. The latter do not refine the variety, this one being already finite, but their prediction allows to speed up the computation. 

\begin{proposition}\label{prop:Uj}
   The vector space $\code{U}_{n-2, n-1} \eqdef (\code{V}_{n-2}+\code{V}_{n-1})\cap\fq{}[\Xm,\Ym]_{\le 2}$ contains more than $ m$ linearly independent polynomials of degree $2$, that are linear combination of the terms $Y_j(2X_j-1)$ for $j\in\Iintv{k+1}{ n-1}$. We denote by
   $\cU_{n-2,n-1}$ this set of $m$ polynomials.

   Moreover, $\code{V}_{n-2}+\code{V}_{n-1}$ contains an additional
   polynomial of degree $2$ expressing the monomial 1 in terms of the
   $Y_j(2X_j-1)$ for $j\in\Iintv{k+1}{ n-1}$, denoted  by $u_{n-2,n-1}$.

   If the characteristic of the field is $2$, the $m+1$ polynomials
   are linear in $Y_j$ and $1$. The first $m$ polynomials in
   ${U}_{n-2,n-1}$ already belong to $\code{V}_{n}$, but the additional
   polynomial $u_{n-2,n-1}$ generates together with $\cV_n$ a vector
   space of dimension $2m$ (hence one more linear polynomial than the
   ones in $\cV_n$).
\end{proposition}
\begin{proof}
  The terms appearing in the polynomials in $\cV_{n-2}$ are
  $1, Y_{n-1}$ and $Y_jX_j^2$ for $j\in\Iintv{k+1}{ n-3}$. The ones in
  $\cV_{n-1}$ are $1, Y_{n-2}$ and $Y_{j}X_j^2 + Y_j(1-2X_j)$,
  $j\in\Iintv{k+1}{ n-3}$.  This means that the polynomials in
  $\code{U}_{n-2,n-1}$ can all be expressed as linear combination of
  the $3m-3$ monomials $Y_jX_j^2$ of degree 3, the monomial $1$ of
  degree $0$, plus $3m-1$ terms of degree at most $ 2$: the monomials $Y_{n-1}$,
  $Y_{n-2}$ and the $Y_j(2X_j-1)$'s, $j\in\Iintv{k+1}{ n-3}$. The dimension
  of the vector space $\code{U}_{n-2,n-1}$ is $4m-2$, so that we get
  at least $m$ linearly independent polynomials of degree $2$ in
  $\code{U}_{n-2, n-1}$ that are combination of $3m-1$ terms
  $(Y_{k+1}(2X_{k+1}-1),\dots,Y_{n-3}(2X_{n-3}-1),Y_{n-2},Y_{n-1})$.
  If the characteristic of the field is 2, then we get $m$ linear
  polynomials in $Y$, that are already in $\code{V}_n$ according to
  Proposition~\ref{prop:dimalternant}.

  In all cases, we also get an additional polynomial of degree 2
  involving these $3m-1$ terms and the monomial 1, and this gives in
  characteristic 2 an additional affine linear polynomial in $Y$.
\end{proof}

This can be generalized to the following vector spaces. For any
$k+1\le j_1 < j_2 \le n-1$, define the vector space
\begin{align}
  \code{U}_{j_1,j_2} &= \frac1{X_{j_1}-X_{j_2}}\left((\code{V}_{j_1}+\code{V}_{j_2})\cap (X_{j_1}-X_{j_2})\cdot \fq{}[\Xm,\Ym]_{\le 2}\right)
\end{align}
that consists of the polynomials $p$ such that
$p(X_{j_1}-X_{j_2})\in\code{V}_{j_1}+\code{V}_{j_2}$.
\begin{proposition} For any $k+1\le j_1 < j_2 \le n-1$, $(j_1,j_2)\ne(n-2,n-1)$, we have
  $\dim(\code{U}_{j_1,j_2})\ge m$, and the polynomials in $\code{U}_{j_1,j_2}$ are linear combination of the following terms:
  \begin{align*}
      Y_j(2X_{j}-X_{j_1}-X_{j_2}), && j\in\Iintv{k+1}{ n-1}.
  \end{align*}
\end{proposition}
\begin{proof}
  The $2m-1$ polynomials in $\cV_{j_1}$ contains the following $3m-1$ terms:
  \begin{align*}
    \begin{cases}
      Y_j(X_j-X_{j_1})^2 & j\in\Iintv{k+1}{ n-1}, j\ne j_1\\
      1 & 
    \end{cases}
  \end{align*}
  It is the same for $\cV_{j_2}$, but we can rewrite, for $j\in\Iintv{k+1}{ n-1}, j\ne j_1,j_2$:
  \begin{align*}
  Y_j(X_j-X_{j_2})^2 = Y_j(X_j-X_{j_1})^2 + Y_j(2X_{j}-X_{j_1}-X_{j_2})(X_{j_1}-X_{j_2}),
  \end{align*}
  so that the $4m-2$ polynomials generating $\code{V}_{j_1}+\code{V}_{j_2}$ can be written in terms of the following terms:
  \begin{align*}
    \begin{cases}
      Y_j(X_j-X_{j_1})^2 & j\in\Iintv{k+1}{ n-1}, j\ne j_1,j_2,\\
      1 & \\
      Y_{j_1}(X_{j_1}-X_{j_2})^2 &\\
      Y_{j_2}(X_{j_1}-X_{j_2})^2& \\
      Y_j(2X_{j}-X_{j_1}-X_{j_2})(X_{j_1}-X_{j_2}), & j\in\Iintv{k+1}{ n-1}, j\ne j_1,j_2.
    \end{cases}
  \end{align*}
  If we eliminate the $3m-2$ first terms that are not multiple of
  $X_{j_1}-X_{j_2}$, we get at least $m$ linearly independent
  polynomials that are multiple of $X_{j_1}-X_{j_2}$. After division
  by $X_{j_1}-X_{j_2}$, the polynomials are linear combination of the
  $3m-1$ terms $Y_j(2X_j-X_{j_1}-X_{j_2})$ for $j\in\Iintv{k+1}{n-1}$.
\end{proof}
\begin{remark}
  If the characteristic of the field is 2, then we can divide directly
  by $X_{j_1}^2+X_{j_2}^2$ to get $m$ linear polynomials in the
  $Y_j$'s. As we cannot have more than $2m-1$ homogeneous linear
  polynomials in the $Y_j$'s if the multipliers $y_j$ generate $\fqm$
  over $\fq$, then all the new polynomials are linearly dependent from
  the previous ones. 
\end{remark}
Despite the existence of a quadratic number of vector spaces $\cU_{j_1,j_2}$, in practice  it is sufficient to exploit the polynomials derived from the $3m-2$ subspaces $\cU_{\ell,n-2},\; \ell \in \Iintv{k+1}{n-1}\setminus \{n-2\}$, thus reducing the complexity of this step.

\begin{fact}\label{fact:dimU}
  Experimentally, $\dim(\code{U}_{j_1,j_2})=m$ and for odd $q$,
  \begin{align*}
\dim_{\Fq}(\cup_{\ell\in\Iintv{k+1}{n-1}, \ell\ne n-2}\code{U}_{\ell,n-2}) &= m(3m-2).
  \end{align*}
\end{fact}

\subsubsection*{Step \eqref{step:4}: eliminating $2m-1$ variables
  $Y_j$ using the linear polynomials from Step~\eqref{step:1}} We
consider the case $q$ odd. Assuming Fact~\ref{fact:dimU}, the system
$\oplus_{j_1\in\Iintv{k+1}{n-3}} {\cU}_{j_1,n-2}$ contains $m(3m-2)$
linearly independent polynomials, and they can all be expressed as
linear combination of the monomials $Y_j, j\in \Iintv{k+1}{n-1}$ and
$Y_{j}X_{j'}$ for $j\in\Iintv{k+1}{n-1}$ and
$j'\in\Iintv{k+1}{n-3}$.

The system $\cV_n$ contains $2m-1$ homogeneous linear polynomials in
the $Y_j$'s, expressing the $Y_i$'s for $i\notin I$ in term of the
$Y_{\ell}$'s for $\ell\in I$. If we use them to eliminate the $2m-1$
variables $Y_i$'s ($i\notin I$) from the polynomials in
$\oplus_{j_1\in\Iintv{k+1}{n-3}} {\cU}_{j_1,n-2}$, we are left with
polynomials that are linear combination of $m$ linear monomials
$\lbrace Y_\ell : \ell \in I\rbrace$, and $m(3m-3)$ quadratic
monomials $Y_{j}X_{j'}$ for $j\in I$ and
$j'\in\Iintv{k+1}{n-3}$. This means that we have as many polynomials as
monomials. However, the polynomials now have no reason to remain
linearly independent, and in fact they are not.

Experimentally, after linearization, we get one polynomial expressing
each quadratic term $Y_j X_{j'}$ in terms of the $m$ independent
$\lbrace Y_\ell : \ell \in I\rbrace$ and $m$ reductions to zero, as
we cannot get more than $2m-1$ linear polynomials relating the
$Y_j$'s: a basis $\cU$ of $\oplus_{j_1}\code{U}_{j_1,n-2}$ modulo $\cV_n$ has the
shape
\begin{align*}
\cU \eqdef \left\{  Y_{j}X_{j'} + f_{j,j'}(Y_{\ell}: \ell \in I)\mid j\in I, j' \in \Iintv{k+1}{n-3}\right\}.
\end{align*}

We can now use the polynomial $u_{n-2,n-1}$ from
Proposition~\ref{prop:Uj}, that is a linear combination of the
monomials $1, Y_{n-2}, Y_{n-1}$ and the $Y_j(2X_j-1)$ for
$j\in\Iintv{k+1}{n-3}$. We eliminate the $Y_j$'s, $j\notin I$ using
equations in $\cV_n$ and the $Y_jX_{j'}$ for
$j\in I, j'\in\Iintv{k+1}{n-3}$ using $\cU$ and obtain a linear
polynomials in the $Y_j$'s and $1$. Note that, as we already have
$2m-1$ homogeneous polynomials between the $Y_j$'s, we cannot have
another homogeneous polynomials, hence the polynomials contains the
constant 1.

To perform the elimination and the linearization, we can perform
linear algebra on a matrix where the columns are the $Y_{j}X_{j'}$,
hence $\bigO{m^2}$ columns, and the rows are the basis for
$\code{U}_{j,n-2}$ and the $X_iL_j$ with $L_j$ a linear polynomial in
$\code{V}_n$. This makes $\bigO{m^2}$ rows, and a complexity in
$\bigO{m^{2\omega}}$.
\subsubsection*{Step \eqref{step:5}: computing   linear polynomials for the $X$ variables}
\begin{proposition}\label{prop:linX}
Assume that a basis of $\cup_{j\ne n-2} \code{U}_{j,n-2}$ where the linear polynomials from $\cV_n$ have been eliminated is given by
\begin{align}\label{eq:cupUj}
  \begin{cases}
    Y_j X_{j'} + f_{j,j'}(Y_\ell : \ell \in I, 1), &\text{ for all } j\in I, j'\in\Iintv{k+1}{ n-3},\\
    \sum_{j\in I} a_jY_j - 1, & \text{ with } a_j\in\fq.
  \end{cases}
\end{align}
Then  the vector space generated by the polynomials \eqref{eq:cupUj} contains the polynomials
\begin{align}
  \label{eq:linearX}
  X_{j'} + \sum_{j\in I} a_jf_{j,j'}(Y_\ell : \ell \in I, 1),&& j'\in\Iintv{k+1}{n-3}.
\end{align}
\end{proposition}
\begin{proof} We have
\begin{align*}
  \sum_{j\in I} a_j(Y_jX_{j'}+f_{j,j'}(Y_\ell : \ell \in I, 1)) &= X_{j'} +X_{j'}( \sum_{j\in I} a_jY_j - 1)  + \sum_{j\in I} a_jf_{j,j'}(Y_\ell : \ell \in I, 1)
\end{align*}
so that we get in the ideal generated by~\eqref{eq:cupUj} one affine linear polynomial expressing each $X_{j'}$ in terms of the $\lbrace Y_\ell : \ell \in I, 1\rbrace$.  
\end{proof}
\subsubsection*{Step \eqref{step:6}: the final Gröbner basis}
Now, if we use the polynomials in~\eqref{eq:linearX} to eliminate the
$X_i$'s from the polynomials in~\eqref{eq:cupUj}, we get one linear
polynomial for each term of degree 2 in $\Ym$. Let $I_1$ be the set $I$
minus one element $i\in I$ such that $a_i\ne 0$.  The final basis has
the shape
\begin{align*}
  \begin{cases}
  Y_jY_{j'} + L'_{j,j'}(Y_\ell: \ell \in I_1, 1),&  j, j'\in I_1, j<j',\\
  X_{j'} + f_{j'}(Y_\ell: \ell \in I_1, 1),&  j' \in \Iintv{ k+1}{n-3},\\
  Y_j + L_{j,n}(Y_\ell: \ell \in I_1),&  j \notin I,\\
  Y_i + L_{i,n}(Y_\ell : \ell \in I_1, 1), & i \in I \setminus I_1.
\end{cases}
\end{align*}
This describes a variety of dimension 0 with $m$ solutions (as
$\#I_1=m-1$, only $m$ monomials 1 and $Y_i : i \in I_1$ are not
leading term of a polynomial in the ideal), that are exactly the $m$
solutions obtained by applying the Frobenius morphism. This proves
that the basis is a Gröbner basis. It coincides with the basis that
would have been computed from $\code{S'}$ plus the cubic polynomials
from the $\cV_j$'s. However, our approach is more efficient, because it
avoids unnecessary calculations.

\subsection{The particular  $q=2$ case.}
\label{sec:q=2}
Assumption~\ref{assumption: rank q=2} asserts that the rank of $\cS'$ is
smaller than for all the other cases, namely
$\rank(\cS') = \binom{3m}{2}-3m$ instead of $\binom{3m}{2}-m$. This
invalidates all the combinatorial arguments for the dimensions of
$\cV_j$'s and for the number of degree falls. In this case we have
\[
\dim(\cV_j)=m-1,
\]
In particular for $j=n$,  the number of independent linear polynomials in $\Ym$ in a basis of $\code{S'}$ is $m-1$. This result is rather technical, its essence not being merely combinatorial, and is beyond the scope of this article.

More interestingly, for the specific $q=2$ case, we can use the following particular identity from~\cite{FGOPT13}:
\begin{align}
  \label{eq:q=2}
(Y_i X_i^2)(Y_i)^2=(Y_i)(Y_iX_i)^2.
\end{align}
Replacing as before the terms $Y_iX_i^a$ for $a\in\Iintv{0}{2}$ by their values from \eqref{eq: systematic_form} and, as $q=2$, the terms $(Y_iX_i^a)^2$ for $a\in\Iintv{0}{1}$ by $
(Y_i X_i^a)^2 = -\sum_{j=k+1}^n p_{i,j}Y_j^2X_j^{2a}
$
leads, after expansion, to the system
\begin{align}
  \label{eq:Sq=2}
  \cB \eqdef \left\{\sum_{k+1\le j<j'\le n} p_{i,j}p_{i,j'} Y_j Y_{j'}(Y_{j}+Y_{j'}) (X_{j}+X_{j'})^2=0 \mid i \in \Iintv{1}{k} \right\}.
\end{align}
\begin{proposition}
  For $r=3$, $q=2$ with the specialization \eqref{eq:specialization}, the system $\cB$ becomes
  \begin{align}\label{eq:Sprimeq=2}
    \cB' = \left\{\sum_{\substack{j,j'=k+1\\j<j'}}^{n-1} p_{i,j}p_{i,j'} Y_j Y_{j'}(Y_{j}+Y_{j'}) (X_{j}+X_{j'})^2 + \sum_{j=k+1}^{n-1} p_{i,j}p_{i,n}Y_j^2 \mid i \in \Iintv{1}{k} \right\}.
\end{align}
\end{proposition}
\begin{proof}
For $q=2$ we use the identity $(Y_iX_i^2)(Y_i)^2=(Y_i)(Y_iX_i)^2$, and the fact that for $q=2$, $Y_i^2=   \sum_{k+1\le j\le n-1} p_{i,j} Y_j^2$ and $(Y_i X_i)^2 = \sum_{k+1\le j\le n-1} p_{i,j} Y_j^2 X_j^2$. After expansion, using the fact that $Y_jX_{j'}^2+Y_{j'}X_{j'}^2+Y_{j'}X_{j}^2+Y_jX_j^2=(Y_j+Y_{j'})(X_j^2+X_{j'}^2)$ for any $k+1\le j < j'\le n-1$, we get $\cB'$.
\end{proof}
As before, we call $\code{B}'$ the vector space generated by the polynomials $\cB'$. The system $\cB'$ contains $\binom{3m}{2}-3m$ linearly independent polynomials of bidegree $(2,3)$ in the two blocks of variables $\Xm$
and $\Ym$ and its linearization involves $\binom{3m}{2}$ variables
\begin{align}
\begin{cases}
          Y_j Y_{j'}(Y_j+Y_{j'}) (X_{j}^2+X_{j'}^2) & \text{ for } k+1\le j < j' \le n-1,\\
          Y_j^2 & \text{ for } j'=n \text{ and } X_{n}=\infty.
        \end{cases}
\end{align}

In a similar way to what done for $\cS'$, we can define the vector spaces
\begin{align}
  \code{E}_j &= \frac1{Y_j}\left(\code{B'}\cap Y_j\cdot \fq{}[\Xm,\Ym]_{\le 3}\right), & j\in\Iintv{k+1}{n-1}
\end{align}
and an associated basis $\cE_j$ of $\code{E}_j$. The following proposition is fundamental for the efficiency of the attack.
\begin{proposition}
  For each $j\in\Iintv{k+1}{n-1}$,  for each polynomial
  \begin{align*}
    \sum_{{j'} \in \Iintv{k+1}{n}\setminus \{j\}} v_{j'} Y_{j'} (X_{j'}+X_{j})^2 \in \cV_j
  \end{align*}
  the polynomial
  \begin{align*}
    \sum_{{j'} \in \Iintv{k+1}{n}\setminus \{j\}} v_{j'} Y_{j'} (X_{j'}+X_{j})
  \end{align*}
  cancels on the solution $\xv,\yv$.  This produces new bilinear
  polynomials that we can add to our algebraic system.
\end{proposition}
\begin{proof}
Comparing the polynomials in $\cS'$ \eqref{eq:Sprime} and $\cB'$ \eqref{eq:Sprimeq=2}, it is clear that 
\[
\sum_{{j'} \in \Iintv{k+1}{n}\setminus \{j\}} v_{j'} Y_{j'} (X_{j'}+X_{j})^2 \in \cV_j \iff \sum_{{j'} \in \Iintv{k+1}{n}\setminus \{j\}} v_{j'} Y_{j'} (Y_{j'}+Y_{j})(X_{j'}+X_{j})^2 \in \cE_{j},
\]
hence $\dim(\cE_{j})=m-1.$ We can split a polynomial in $\cE_{j}$ in the following way:
\begin{align*}
&\sum_{{j'} \in \Iintv{k+1}{n}\setminus \{j\}} v_{j'} Y_{j'} (Y_{j'}+Y_{j})(X_{j'}+X_{j})^2\\
=&\sum_{{j'} \in \Iintv{k+1}{n}\setminus \{j\}} v_{j'} Y_{j'}^2 (X_{j'}+X_{j})^2+Y_j\sum_{{j'} \in \Iintv{k+1}{n}\setminus \{j\}} v_{j'} Y_{j'} (X_{j'}+X_{j})^2.
\end{align*}
Since $Y_j\sum_{{j'} \in \Iintv{k+1}{n}\setminus \{j\}} v_{j'} Y_{j'} (X_{j'}+X_{j})^2$ is in the ideal generated by $\cV_j$, we obtain
\[
\sum_{{j'} \in \Iintv{k+1}{n}\setminus \{j\}} v_{j'} Y_{j'}^2 (X_{j'}+X_{j})^2 \in \cE_j+Y_j\cV_j.
\]
Since the coefficients $v_{j'} \in \F_2$, by applying the Frobenius map $m-1$ times, we get that the polynomial
\begin{align*}
\left(\sum_{{j'} \in \Iintv{k+1}{n}\setminus \{j\}} v_{j'} Y_{j'}^2 (X_{j'}+X_{j})^2=0\right)^{2^{m-1}}=\sum_{{j'} \in \Iintv{k+1}{n}\setminus \{j\}} v_{j'} Y_{j'} (X_{j'}+X_{j})
\end{align*}
belong to the ideal generated by $\cE_j+\cV_j$ and the polynomials
$Y_{j'}^{2^m}-Y_{j'}$, $X_{j'}^{2^m}-X_{j'}$ for $j'\in\Iintv{k+1}{n}$, hence cancels on the
solution.
\end{proof}
As a consequence, from the computation of the $\cV_j$'s, and under the assumption that each $\cV_j$ contains $m-1$ linearly independent polynomials, we can produce  $(m-1)(3m-1)$ bilinear polynomials that are independent, as well as $m-1$ linearly independent linear polynomials in $Y_j$'s (in $\cV_n$). A Gr\"obner basis of the system formed by those polynomials together with the polynomials in the systems $\cV_j$'s can then be computed, and experimentally all computations are done by staying at degree $3$, which gives a global complexity that is polynomial in $m$. The last steps of the attack (computation of the solutions using the FGLM algorithm, and recovering of the entier support and multiplier) are identical to the $q>2$ case.

\subsection{Interlacing the algebraic recovering with the filtration}
\label{sec:interlacing}
We now get back to distinguish between the full-length vectors $\xv$ and $\yv$ and their shortening due to the filtration attack. We can restore the information lost from the filtration shortening, by simply repeating the attack twice on different sets. This is possible because $2(r-3)$ is very small compared to $n$. Indeed, if we shorten the positions corresponding to $\Ical_1\eqdef \Iintv{1}{r-3}$ (the order is irrelevant) during the filtration attack, at the end of the algebraic recovering we have access to $m$ pairs of vectors
\[
\bar{\xv}_{\Ical_1} \quad \text{and} \quad \bar{\yv}_{\Ical_1} \left(\prod_{i \in \Ical_1} (\bar{\xv}_{\Ical_1}-x_i)\right).
\]
Analogously, if we shorten the positions corresponding to $\Ical_2\eqdef \Iintv{(r-3)+1}{2(r-3)}$ during the filtration attack, at the end of the algebraic recovering we have access to $m$ pairs of vectors
\[
\bar{\xv}_{\Ical_2} \quad \text{and} \quad \bar{\yv}_{\Ical_2} \left(\prod_{i \in \Ical_2} (\bar{\xv}_{\Ical_2}-x_i)\right).
\]
In particular, if the same specialization has been chosen, we can couple $m$ pairs $(\bar{\xv}_{\Ical_1}, \bar{\xv}_{\Ical_2})$ such that the two vectors of each pair coincide on the last $n-(r-3)$ coordinates. We can easily detect them from the last $3m$ coordinates, so that we do not need to solve $2m$ linear systems but it is sufficient to choose one pair and solve only the 2 corresponding linear systems. In this way we obtain a full solution $\bar{\xv}$ for the original problem as
\[
\bar{\xv} = (\underbrace{\bar{x}_1,\dots,\bar{x}_{r-3}}_{\substack{\text{first } r-3\\ \text{ coordinates of } \bar{\xv}_{\Ical_1}}}, \underbrace{\bar{x}_{r-2},\dots,\bar{x}_{2(r-3)}}_{\substack{\text{first } r-3 \\\text{ coordinates of } \bar{\xv}_{\Ical_2}}},\underbrace{\bar{x}_{2(r-3)+1},\dots,\bar{x}_{n-3},0,1,\infty}_{\substack{\text{last common coordinates  of}\\ \bar{\xv}_{\Ical_1}\text{ and }\bar{\xv}_{\Ical_2}}}).
\]
By replacing the found values in the corresponding $\bar{\yv}_{\Ical_1} \left(\prod_{i \in \Ical_1} (\bar{\xv}_{\Ical_1}-x_i)\right)$ and $\bar{\yv}_{\Ical_2} \left(\prod_{i \in \Ical_2} (\bar{\xv}_{\Ical_2}-x_i)\right)$, we retrieve  $\bar{\yv}_{\Ical_1}$ and $\bar{\yv}_{\Ical_2}$. Similarly to what done for the support, we can put together the information of these two vectors and get
\[
\bar{\yv} = (\underbrace{\bar{y}_1,\dots,\bar{y}_{r-3}}_{\substack{\text{first }r-3\\\text{ coordinates of } \bar{\yv}_{\Ical_1}}}, \underbrace{\bar{y}_{r-2},\dots,\bar{y}_{2(r-3)}}_{\substack{\text{first }r-3\\\text{ coordinates of }\bar{\yv}_{\Ical_2}}},\underbrace{\bar{y}_{2(r-3)+1},\dots,\bar{y}_{n-1},1}_{\substack{\text{last common coordinates of }\\ \bar{\yv}_{\Ical_1}\text{ and }\bar{\yv}_{\Ical_2}}}).
\]

So, a pair of valid support and multiplier has been recovered. However, $\bar{\xv} \notin \Fqm^n$, because $\bar{x}_n= \infty$. The last question is therefore how to get a valid pair of support and multiplier such that both are defined over $\Fqm$, \ie how to get the alternant representation. In other words, we need to determine some $f \in GL_2(\Fqm)$ and $\lambda \in \Fqm\setminus\{0\}$ such that 
\[\bar{x}_i' = f(\bar{x}_i) \in \Fqm,\quad \forall i \in \Iintv{1}{n}\]
and
\[\bar{y}_i' = \lambda \theta(f,\bar{x}_i)^{r-1} \bar{y}_i \in \Fqm,\quad \forall i \in \Iintv{1}{n}.\]
We observe that, since there are only $n-1$ coordinates of $\bar{\xv}$ in $\Fqm$ and $n-1<q^m$, there exists at least one element $\hat{x} \in \Fqm$ that is different from all $\bar{\xv}$ coordinates. We also remark that $\hat{x} \ne 0$, since $\bar{x}_{n-2} = 0$, so the map $f$ on $\bar{\F}_{q^m}$
\[
f\eqdef\frac{z}{z-\hat{x}}
\]
is induced by an element of the linear group. We have $\theta(f, z)= z- \hat{x}$ if $z \in \Fqm$ and $\theta(f, \infty) = 1$ and we choose $\lambda=1 $. Therefore 
\begin{align*}
\bar{x}_i' = \frac{\bar{x}_i}{\bar{x}_i-\hat{x}},&\quad i \in \Iintv{1}{n-1},\\
\bar{x}_n' = 1,&\\
\bar{y}_i' =  (\bar{x}_i- \hat{x})^{r-1} \bar{y}_i,&\quad i \in \Iintv{1}{n-1},\\
\bar{y}_n' = \bar{y}_n=1.&\\
\end{align*}
We finally obtained a support and a multiplier with coordinates over $\Fqm$ that define the public code.
This concludes the key-recovery attack on high-rate random alternant codes.

 \section{Conclusion}
\subsection*{Breaking the $m=2$ barrier.} Even if the first step of the attack, namely the derivation of an alternant code of degree $3$ from the original 
alternant code is reminiscent of the filtration attack \cite{COT14,COT17} which allowed to break a McEliece scheme based on wild Goppa codes of extension degree $2$, it differs in a crucial way. 
In \cite{COT14,COT17} building this decreasing sequence of codes is based on taking the shortening of the original Goppa code and square codes of it. This method is bound to fail
when the extension degree $m$ is greater than $2$ as explained in \cite[\S III]{COT17}. This is a pity, because the square code is much easier to understand in this case since it retains the polynomial
structure of the generalized Reed-Solomon super-code. To go beyond the $m=2$ barrier, it is mandatory to look for the the square of shortenings of the dual of the Goppa/alternant code instead. This square is much more difficult to understand, since taking the dual of a subfield subcode as is the case here yields by Delsarte's theorem the trace code of a generalized Reed-Solomon 
code which loses a lot of its polynomial description. This is certainly one of the main difficulties that had to be overcome to transform the distinguisher of \cite{FGOPT11} into an attack. 

This is precisely what is achieved here. Besides shedding some light (together with \cite{MT22}) on the structure of the square of such dual codes, this paper also introduces some other crucial ideas. The first one being the crucial role of being a generic alternant code rather than a Goppa code and then the condition $r \geq q+1$ which is crucial for our main result, Theorem \ref{thm:main_filtration} to hold.
It is rather surprising that right now it seems that this filtration strategy needs apparently additional ideas to make it work in the particular case of Goppa codes as observed in Subsection \ref{ss:Goppa} despite the 
fact that Goppa codes are even more structured than plain alternant codes. However there are clearly reasons to believe that this paper has solved the main issue when it comes to transform the distinguisher of 
\cite{FGOPT11} into an attack, which is to be able to work with the dual of the alternant/Goppa code, since this is the crux in breaking the $m=2$ barrier. Therefore it is tempting to conjecture that Goppa codes,  at least in the regime where they are distinguishable from random codes (which applies in particular to the CFS scheme \cite{CFS01}) should eventually be attacked in polynomial by some variation the attack that has been given here. 

\subsection*{Understanding the Gröbner basis approach.}
The filtration approach which amounts here to construct a nested sequence of dual of alternant codes of decreasing degree results in a final alternant code of degree $3$ in the case at hand.
From here, a dedicated Gröbner basis is used to recover the polynomial structure of the alternant code. This approach is able to take into account in a very efficient way the fact that the 
multiplier and the support that are sought should satisfy certain non standard constraints (all the entries of the multiplier are non zero whereas all the entries of the support are distinct). Taking these constraints 
into account as we do here results in speeding up significantly the algebraic system solving by adding many new low degree equations to the algebraic system. When the code is a Goppa code it turns out that solving the relevant system behaves a little bit differently (as for the case of the filtration). However, unlike the case of the filtration where right now this part of the attack does not work at all, it seems that here even if the Gröbner basis computation consists of more
steps, solving the whole system should still be polynomial. Actually there are reasons to conjecture that the Gröbner basis approach stays of polynomial complexity for recovering the polynomial structure of an alternant code of degree $4$ and even potentially beyond this, say up to a constant degree. More generally, this raises the issue of  understanding  how the complexity of recovering the polynomial structure of a degree $r$ alternant/Goppa code through this Gröbner basis approach scales with $r$. An obvious application of this study would be to be able to increase the values of the field size
$q$ for which we still have a polynomial attack. This could also allow to get a better understanding how the McEliece resists to this algebraic approach in general and is desirable in the light of the ongoing NIST competition where the McEliece cryptosystem is an alternate fourth round candidate.

\newcommand{\etalchar}[1]{$^{#1}$}

\appendix
\section{Further results about trace codes}\label{sec:extended}
We will state and prove here a simple result on the extension of trace codes which will be useful in our context.
\begin{prop} \label{prop:trace_fqm}
	Let $\CC \subseteq \Fqm^n$ be an $\Fqm$- linear code. Then
\[
	\Tr{\CC}_{\Fqm}=\sum_{i=0}^{m-1} \CC^{q^i}.
\]
\end{prop}
\begin{proof}
Take any $\cv \in \CC$. Then $\Tr{\cv} = \cv + \cv^q + \cdots + \cv^{q^{m-1}}$ also belongs to
$\CC + \CC^q + \cdots + \CC^{q^{m-1}}$. This proves that $\Tr{\CC}\subseteq \sum_{i=0}^{m-1} \CC^{q^i}
$ and therefore $\Tr{\CC}_{\Fqm} \subseteq \sum_{i=0}^{m-1} \CC^{q^i}$. 
On the other hand, let us prove that any $\CC^{q^i}$ is a subspace of
$\Tr{\CC}_{\Fqm}$ for any $i$.
Consider an arbitrary $\Fq$-basis $\alpha \eqdef \{\alpha_1,\cdots,\alpha_m\}$ of $\Fqm$. Let $\xv_i \eqdef 
\Tr{\alpha_i \cv}$. Since 
$$\xv_i = \alpha_i \cv + \alpha_i^q \cv^q + \cdots + \alpha_i^{q^{m-1}}\cv^{q^{m-1}}$$
we have that
\begin{eqnarray*}
\begin{pmatrix} \xv_1 & \xv_2 & \cdots & \xv_m \end{pmatrix} & = & 
\begin{pmatrix} \cv & \cv^q & \cdots & \cv^{q^{m-1}} \end{pmatrix}
\underbrace{\begin{pmatrix} \alpha_1 & \alpha_2 & \cdots & \alpha_m \\
\alpha_1^q & \alpha_2^q & \cdots & \alpha_m^{q} \\
\vdots & \vdots & \vdots & \vdots \\
\alpha_1^{q^{m-1}} & \alpha_2^{q^{m-1}} & \cdots & \alpha_m^{q^{m-1}}
\end{pmatrix}}_{\eqdef \Mm(\alpha)}
\end{eqnarray*}
$\Mm(\alpha)$ is the Moore matrix associated to $\{\alpha_1,\cdots,\alpha_m\}$ and is  invertible because the 
$\alpha_i$'s are linearly independent over $\Fq$. Therefore
$$
\begin{pmatrix} \cv & \cv^q & \cdots & \cv^{q^{m-1}} \end{pmatrix} = \begin{pmatrix} \xv_1 & \xv_2 & \cdots & \xv_m \end{pmatrix} \Mm(\alpha)^{-1} 
$$
and therefore all the $\cv^{q^i}$ are $\Fqm$-linear combinations of the $\xv_j$'s and belong therefore to 
$\Tr{\CC}_{\Fqm}$. This shows that $\CC^{q^i} \subseteq \Tr{\CC}_{\Fqm}$ for any $i$ and shows therefore the reverse inclusion 
$$
\CC + \CC^q + \cdots + \CC^{q^{m-1}} \subseteq \Tr{\CC}_{\Fqm}.
$$
\end{proof}

\section{Proofs of Subsection \ref{ss:Goppa}}

Let us recall the proposition we are going to prove.
\propCodes*

\begin{proof}
\noindent
{\bf Proof of \eqref{eq:altrpu}.}
\eqref{eq:altrpu} was proved in \cite[Prop. 1]{B00}.

\medskip
\noindent
{\bf Proof of \eqref{eq:shi}.}

Choose an $\Fq$-basis $\{\alpha_1,\cdots,\alpha_m\}$ of $\Fqm$. We are first going to prove that
\begin{equation}
\label{eq:evaluateto0}
\Goppa{\puv{\xv}{i}}{\Gamma}^\perp + \Fqspan{\mathbf{1}} =\Alt{r}{\puv{\xv}{i}}{\puv{i}{\yv}(\puv{\xv}{i}-x_i)}^\perp+ \Fqspan{\mathbf{1}}.
\end{equation}
		\begin{align*}
			\Alt{r}{\puv{\xv}{i}}{\puv{\yv}{i}(\puv{\xv}{i}-x_i)}^\perp &= \Fqspan{\Tr{\alpha_j \puv{\xv}{i}^a (\puv{\xv}{i}-x_i)\puv{\yv}{i}} \mid a \in \Iintv{0}{r-1}, j \in \Iintv{0}{m-1}}\\
			&=\Fqspan{\Tr{\alpha_j\puv{\xv}{i}^{a+1} \puv{\yv}{i}}-\Tr{\alpha_j x_i \puv{\xv}{i}^a \puv{\yv}{i}} \mid a \in \Iintv{0}{r-1}, j \in \Iintv{0}{m-1}}\\
			&\subseteq\Fqspan{\Tr{\alpha_j\puv{\xv}{i}^b \puv{\yv}{i}} \mid b \in \Iintv{0}{r}, j \in \Iintv{0}{m-1}}\\
			&=\Alt{r+1}{\puv{\xv}{i}}{\puv{\yv}{i}}^\perp\\
			&= \Goppa{\puv{\xv}{i}}{\Gamma}^\perp + \Fqspan{\mathbf{1}}
		\end{align*}
	
		On the other hand, since $\Gamma(x_i)\neq 0$, then $\Gamma(\puv{\xv}{i}) \not\in \Fqspan{\puv{\xv}{i}^a (\puv{\xv}{i}-x_i) \mid a \in \Iintv{0}{r-1}}$. Therefore $\Fqspan{\Gamma(\puv{\xv}{i}), \puv{\xv}{i}^a (\puv{\xv}{i}-x_i) \mid a \in \Iintv{0}{r-1}}$ is a vector space of dimension $r+1$ of evaluations of polynomials with degree at most $r$. Hence
		\[\Fqspan{\Gamma(\puv{\xv}{i}), \puv{\xv}{i}^a (\puv{\xv}{i}-x_i) \mid a \in \Iintv{0}{r-1}}=\Fqspan{\puv{\xv}{i}^b \mid b \in \Iintv{0}{r}}.\]
		Since $\Tr{\alpha_j \Gamma(\puv{\xv}{i}) \puv{\yv}{i}}=\Tr{\alpha_j \cdot \mathbf{1}}\in \Fqspan{\mathbf{1}}$, we get
		\begin{align*}
		\Goppa{\puv{\xv}{i}}{\Gamma}^\perp + \Fqspan{\mathbf{1}} &=\Fqspan{\Tr{\alpha_j\puv{\xv}{i}^b \puv{\yv}{i}} \mid b \in \Iintv{0}{r}, j \in \Iintv{0}{m-1}}\\
		&\subseteq \Fqspan{\Tr{\alpha_j \puv{\xv}{i}^a (\puv{\xv}{i}-x_i)\puv{\yv}{i}} \mid a \in \Iintv{0}{r-1}, j \in \Iintv{0}{m-1}} + \Fqspan{\mathbf{1}}\\
		&=\Alt{r}{\puv{\xv}{i}}{\puv{\yv}{i}(\puv{\xv}{i}-x_i)}^\perp+ \Fqspan{\mathbf{1}}.
		\end{align*}
Because of the last point, $\sh{i}{\Alt{r+1}{\xv}{\yv}^\perp}$ is the set of codewords of $\Goppa{\xv}{\Gamma} + \Fqspan{\mathbf{1}}$ which 
evaluate to $0$ at position $i$. Clearly the elements of $\Alt{r}{\puv{\xv}{i}}{\puv{\yv}{i}(\puv{\xv}{i}-x_i)}^\perp$ viewed as polynomial evaluations and
extended canonically at position $i$ as the linear space
$$
\left\{\Tr{\yv (\xv-x_i)P(\xv)}: \;P \in \Fqm[X],\;\deg P < r \right\}
$$
 belong to this set. Since the all one vector does not meet this property and because of \eqref{eq:evaluateto0} this implies the point \eqref{eq:shi}.

\noindent
{\bf Proof of \eqref{eq:shicomshj}.} 
This is just a consequence that $\sh{i}{\sh{j}{\CC}}=\sh{j}{\sh{i}{\CC}}$ which holds for any code $\CC$. Here we apply it to $\CC=\Alt{r+1}{\xv}{\yv}^\perp$ and apply the previous point:
\begin{eqnarray*}
\sh{i}{\Alt{r}{\puv{\xv}{j}}{\puv{\yv}{j}(\puv{\xv}{j}-x_j)}^\perp} & = &\sh{i}{\sh{j}{\Alt{r+1}{\xv}{\yv}^\perp}} \;\text{(by \eqref{eq:shi})}\\
& = & \sh{j}{\sh{i}{\Alt{r+1}{\xv}{\yv}^\perp}}\;\;\text{(by the previous remark)}\\
& = & \sh{j}{\Alt{r}{\puv{\xv}{i}}{\puv{\yv}{i}(\puv{\xv}{i}-x_i)}^\perp} \text{(by \eqref{eq:shi})}
\end{eqnarray*}

\end{proof}

\begin{thebibliography}{CGG{\etalchar{+}}14}

\bibitem[ABC{\etalchar{+}}20]{ABCCGLMMMNPPPSSSTW20}
Martin Albrecht, Daniel~J. Bernstein, Tung Chou, Carlos Cid, Jan Gilcher, Tanja
  Lange, Varun Maram, Ingo von Maurich, Rafael Mizoczki, Ruben Niederhagen,
  Edoardo Persichetti, Kenneth Paterson, Christiane Peters, Peter Schwabe,
  Nicolas Sendrier, Jakub Szefer, Cen~Jung Tjhai, Martin Tomlinson, and Wang
  Wen.
\newblock Classic {M}c{E}liece (merger of {Classic McEliece} and {NTS-KEM}).
\newblock \url{https://classic.mceliece.org}, October 2020.
\newblock Third round finalist of the NIST post-quantum cryptography call.

\bibitem[BBB{\etalchar{+}}17]{BBBCDGGHKNNPR17}
Gustavo Banegas, Paulo~S.L.M Barreto, Brice~Odilon Boidje, Pierre-Louis Cayrel,
  Gilbert~Ndollane Dione, Kris Gaj, Cheikh~Thi{\'e}coumba Gueye, Richard
  Haeussler, Jean~Belo Klamti, Ousmane N'diaye, Duc~Tri Nguyen, Edoardo
  Persichetti, and Jefferson~E. Ricardini.
\newblock {D}{A}{G}{S} : Key encapsulation for dyadic {G}{S} codes.
\newblock
  \url{https://csrc.nist.gov/CSRC/media/Projects/Post-Quantum-Cryptography/documents/round-1/submissions/DAGS.zip},
  November 2017.
\newblock First round submission to the NIST post-quantum cryptography call.

\bibitem[BCGO09]{BCGO09}
Thierry~P. Berger, Pierre-Louis Cayrel, Philippe Gaborit, and Ayoub Otmani.
\newblock Reducing key length of the {McEliece} cryptosystem.
\newblock In Bart Preneel, editor, {\em Progress in Cryptology -
  AFRICACRYPT~2009}, volume 5580 of {\em LNCS}, pages 77--97, Gammarth,
  Tunisia, June~21-25 2009.

\bibitem[Ber00]{B00}
Thierry~P. Berger.
\newblock On the cyclicity of {G}oppa codes, parity-check subcodes of {G}oppa
  codes and extended {G}oppa codes.
\newblock {\em Finite Fields Appl.}, 6(3):255--281, 2000.

\bibitem[BJMM12]{BJMM12}
Anja Becker, Antoine Joux, Alexander May, and Alexander Meurer.
\newblock Decoding random binary linear codes in {$2^{n/20}$}: How {$1+1=0$}
  improves information set decoding.
\newblock In {\em Advances in Cryptology - EUROCRYPT~2012}, LNCS. Springer,
  2012.

\bibitem[BLM11]{BLM11}
Paulo Barreto, Richard Lindner, and Rafael Misoczki.
\newblock Monoidic codes in cryptography.
\newblock In {\em Post-Quantum Cryptography~2011}, volume 7071 of {\em LNCS},
  pages 179--199. Springer, 2011.

\bibitem[BLP10]{BLP10}
Daniel~J. Bernstein, Tanja Lange, and Christiane Peters.
\newblock Wild {M}c{E}liece.
\newblock In Alex Biryukov, Guang Gong, and Douglas~R. Stinson, editors, {\em
  Selected Areas in Cryptography}, volume 6544 of {\em LNCS}, pages 143--158,
  2010.

\bibitem[BLP11]{BLP11a}
Daniel~J. Bernstein, Tanja Lange, and Christiane Peters.
\newblock Wild {M}c{E}liece {I}ncognito.
\newblock In Bo-Yin Yang, editor, {\em Post-Quantum Cryptography~2011}, volume
  7071 of {\em LNCS}, pages 244--254. Springer Berlin Heidelberg, 2011.

\bibitem[BM17]{BM17}
Leif Both and Alexander May.
\newblock {Optimizing {BJMM} with Nearest Neighbors: Full Decoding in $2^{2/21
  n}$ and {McEliece} Security}.
\newblock In {\em WCC Workshop on Coding and Cryptography}, September 2017.

\bibitem[CBB{\etalchar{+}}17]{BIGQUAKE}
Alain Couvreur, Magali Bardet, {\'Elise}~Barelli, Olivier Blazy, Rodolfo
  Canto~Torres, Phillipe Gaborit, Ayoub Otmani, Nicolas Sendrier, and
  Jean-Pierre Tillich.
\newblock {BIG QUAKE}.
\newblock \url{https://bigquake.inria.fr}, November 2017.
\newblock {NIST} Round 1 submission for Post-Quantum Cryptography.

\bibitem[CC98]{CC98}
Anne Canteaut and Florent Chabaud.
\newblock A new algorithm for finding minimum-weight words in a linear code:
  Application to {M}c{E}liece's cryptosystem and to narrow-sense {BCH} codes of
  length 511.
\newblock {\em IEEE Trans. Inform. Theory}, 44(1):367--378, 1998.

\bibitem[CCMZ15]{CCMZ15}
Igniacio Cascudo, Ronald Cramer, Diego Mirandola, and Gilles Z{\'e}mor.
\newblock Squares of random linear codes.
\newblock {\em IEEE Trans. Inform. Theory}, 61(3):1159--1173, 3 2015.

\bibitem[CFS01]{CFS01}
Nicolas Courtois, Matthieu Finiasz, and Nicolas Sendrier.
\newblock How to achieve a {McEliece}-based digital signature scheme.
\newblock In {\em Advances in Cryptology - ASIACRYPT~2001}, volume 2248 of {\em
  LNCS}, pages 157--174, Gold Coast, Australia, 2001. Springer.

\bibitem[CGG{\etalchar{+}}14]{CGGOT14}
Alain Couvreur, Philippe Gaborit, Val{\'{e}}rie Gauthier{-}Uma{\~{n}}a, Ayoub
  Otmani, and Jean-Pierre Tillich.
\newblock Distinguisher-based attacks on public-key cryptosystems using
  {Reed-Solomon} codes.
\newblock {\em Des. Codes Cryptogr.}, 73(2):641--666, 2014.

\bibitem[CMCP17]{CMP17}
Alain Couvreur, Irene M{\'a}rquez-Corbella, and Ruud Pellikaan.
\newblock Cryptanalysis of {M}c{E}liece cryptosystem based on algebraic
  geometry codes and their subcodes.
\newblock {\em IEEE Trans. Inform. Theory}, 63(8):5404--5418, 8 2017.

\bibitem[COT14]{COT14}
Alain Couvreur, Ayoub Otmani, and Jean-Pierre Tillich.
\newblock New identities relating wild {G}oppa codes.
\newblock {\em Finite Fields Appl.}, 29:178--197, 2014.

\bibitem[COT17]{COT17}
Alain Couvreur, Ayoub Otmani, and Jean-Pierre Tillich.
\newblock Polynomial time attack on wild {M}c{E}liece over quadratic
  extensions.
\newblock {\em IEEE Trans. Inform. Theory}, 63(1):404--427, 1 2017.

\bibitem[Del75]{D75}
Philippe Delsarte.
\newblock On subfield subcodes of modified {Reed-Solomon} codes.
\newblock {\em IEEE Trans. Inform. Theory}, 21(5):575--576, 1975.

\bibitem[D{\"{u}}r87]{D87}
Arne D{\"{u}}r.
\newblock The automorphism groups of {Reed-Solomon} codes.
\newblock {\em Journal of Combinatorial Theory, Series A}, 44:69--82, 1987.

\bibitem[FGLM93]{FGLM93}
Jean-Charles Faug{\`e}re, Patrizia~M. Gianni, Daniel Lazard, and Teo Mora.
\newblock Efficient computation of zero-dimensional gr{\"o}bner bases by change
  of ordering.
\newblock {\em J. Symbolic Comput.}, 16(4):329--344, 1993.

\bibitem[FGO{\etalchar{+}}11]{FGOPT11}
Jean-Charles Faug{\`e}re, Val{\'e}rie Gauthier, Ayoub Otmani, Ludovic Perret,
  and Jean-Pierre Tillich.
\newblock A distinguisher for high rate {McEliece} cryptosystems.
\newblock In {\em Proc. IEEE Inf. Theory Workshop- ITW~2011}, pages 282--286,
  Paraty, Brasil, October 2011.

\bibitem[FGO{\etalchar{+}}13]{FGOPT13}
Jean-Charles Faug{\`e}re, Val{\'e}rie Gauthier, Ayoub Otmani, Ludovic Perret,
  and Jean-Pierre Tillich.
\newblock A distinguisher for high rate {McEliece} cryptosystems.
\newblock {\em IEEE Trans. Inform. Theory}, 59(10):6830--6844, October 2013.

\bibitem[FOPT10]{FOPT10}
Jean-Charles Faug{\`e}re, Ayoub Otmani, Ludovic Perret, and Jean-Pierre
  Tillich.
\newblock Algebraic cryptanalysis of {McEliece} variants with compact keys.
\newblock In {\em Advances in Cryptology - EUROCRYPT~2010}, volume 6110 of {\em
  LNCS}, pages 279--298, 2010.

\bibitem[GUL09]{GL09}
Val{\'{e}}rie Gauthier-Uma{\~{n}}a and Gregor Leander.
\newblock Practical key recovery attacks on two {McEliece} variants, 2009.
\newblock IACR Cryptology ePrint Archive, Report2009/509.

\bibitem[HP03]{HP03}
W.~Cary Huffman and Vera Pless.
\newblock {\em Fundamentals of error-correcting codes}.
\newblock Cambridge University Press, Cambridge, 2003.

\bibitem[KT17]{KT17a}
Ghazal Kachigar and Jean-Pierre Tillich.
\newblock Quantum information set decoding algorithms.
\newblock In {\em Post-Quantum Cryptography~2017}, volume 10346 of {\em LNCS},
  pages 69--89, Utrecht, The Netherlands, June 2017. Springer.

\bibitem[LS01]{LS01}
Pierre Loidreau and Nicolas Sendrier.
\newblock Weak keys in the {McEliece} public-key cryptosystem.
\newblock {\em IEEE Trans. Inform. Theory}, 47(3):1207--1211, 2001.

\bibitem[MB09]{MB09}
Rafael Misoczki and Paulo Barreto.
\newblock Compact {McEliece} keys from {Goppa} codes.
\newblock In {\em Selected Areas in Cryptography}, Calgary, Canada,
  August~13-14 2009.

\bibitem[McE78]{M78}
Robert~J. McEliece.
\newblock {\em A Public-Key System Based on Algebraic Coding Theory}, pages
  114--116.
\newblock Jet Propulsion Lab, 1978.
\newblock DSN Progress Report 44.

\bibitem[MMT11]{MMT11}
Alexander May, Alexander Meurer, and Enrico Thomae.
\newblock Decoding random linear codes in {$O(2^{0.054n})$}.
\newblock In Dong~Hoon Lee and Xiaoyun Wang, editors, {\em Advances in
  Cryptology - ASIACRYPT~2011}, volume 7073 of {\em LNCS}, pages 107--124.
  Springer, 2011.

\bibitem[MO15]{MO15}
Alexander May and Ilya Ozerov.
\newblock On computing nearest neighbors with applications to decoding of
  binary linear codes.
\newblock In E.~Oswald and M.~Fischlin, editors, {\em Advances in Cryptology -
  EUROCRYPT~2015}, volume 9056 of {\em LNCS}, pages 203--228. Springer, 2015.

\bibitem[MP12]{MP12}
Irene {M{\'a}rquez-Corbella} and Ruud Pellikaan.
\newblock Error-correcting pairs for a public-key cryptosystem.
\newblock CBC 2012, Code-based {C}ryptography {W}orkshop, 2012.
\newblock Available on \url{http://www.win.tue.nl/~ruudp/paper/59.pdf}.

\bibitem[MS86]{MS86}
Florence~J. MacWilliams and Neil J.~A. Sloane.
\newblock {\em The Theory of Error-Correcting Codes}.
\newblock {North--Holland}, Amsterdam, fifth edition, 1986.

\bibitem[MT21]{MT21}
Rocco Mora and Jean-Pierre Tillich.
\newblock On the dimension and structure of the square of the dual of a {Goppa}
  code.
\newblock preprint, 2021.

\bibitem[MT22]{MT22}
Rocco Mora and Jean-Pierre Tillich.
\newblock On the dimension and structure of the square of the dual of a {Goppa}
  code.
\newblock In {\em {WCC 2022 - Workshop on Coding Theory and Cryptography}},
  2022.

\bibitem[Nie86]{N86}
Harald Niederreiter.
\newblock Knapsack-type cryptosystems and algebraic coding theory.
\newblock {\em Problems of Control and Information Theory}, 15(2):159--166,
  1986.

\bibitem[Ran15]{R15}
Hugues Randriambololona.
\newblock On products and powers of linear codes under componentwise
  multiplication.
\newblock In {\em Algorithmic arithmetic, geometry, and coding theory}, volume
  637 of {\em Contemp. Math.}, pages 3--78. Amer. Math. Soc., Providence, RI,
  2015.

\bibitem[RSA78]{RSA78}
Ronald~L. Rivest, Adi Shamir, and Leonard~M. Adleman.
\newblock A method for obtaining digital signatures and public-key
  cryptosystems.
\newblock {\em Commun. ACM}, 21(2):120--126, 1978.

\bibitem[Sen00]{S00}
Nicolas Sendrier.
\newblock Finding the permutation between equivalent linear codes: The support
  splitting algorithm.
\newblock {\em IEEE Trans. Inform. Theory}, 46(4):1193--1203, 2000.

\bibitem[Sho94]{S94a}
Peter~W. Shor.
\newblock Algorithms for quantum computation: Discrete logarithms and
  factoring.
\newblock In S.~Goldwasser, editor, {\em FOCS}, pages 124--134, 1994.

\bibitem[SS92]{SS92}
Vladimir~Michilovich Sidelnikov and S.O. Shestakov.
\newblock On the insecurity of cryptosystems based on generalized
  {Reed-Solomon} codes.
\newblock {\em Discrete Math. Appl.}, 1(4):439--444, 1992.

\bibitem[Ste88]{S88}
Jacques Stern.
\newblock A method for finding codewords of small weight.
\newblock In G.~D. Cohen and J.~Wolfmann, editors, {\em Coding Theory and
  Applications}, volume 388 of {\em LNCS}, pages 106--113. Springer, 1988.

\end{thebibliography}
 \end{document}